%% file: main.tex
\documentclass[conference]{IEEEtran}

\def\BibTeX{{\rm B\kern-.05em{\sc i\kern-.025em b}\kern-.08em
    T\kern-.1667em\lower.7ex\hbox{E}\kern-.125emX}}

\include{macros}

\begin{document}

\title{\ourtitle}

\makeatletter
\newcommand{\linebreakand}{%
  \end{@IEEEauthorhalign}
  \hfill\mbox{}\par
  \mbox{}\hfill\begin{@IEEEauthorhalign}
}
\makeatother

\author{\IEEEauthorblockN{Thomas Humphries}
\IEEEauthorblockA{\textit{University of Waterloo} \\
Waterloo, Canada\\
thomas.humphries@uwaterloo.ca}
\and
\IEEEauthorblockN{Simon Oya}
\IEEEauthorblockA{\textit{University of Waterloo} \\
Waterloo, Canada\\
simon.oya@uwaterloo.ca}
\and
\IEEEauthorblockN{Lindsey Tulloch}
\IEEEauthorblockA{\textit{University of Waterloo} \\
Waterloo, Canada\\
lindsey.tulloch@uwaterloo.ca}
\and
\IEEEauthorblockN{Matthew Rafuse}
\IEEEauthorblockA{\textit{University of Waterloo} \\
Waterloo, Canada\\
matthew.rafuse@uwaterloo.ca}
\linebreakand
\IEEEauthorblockN{Ian Goldberg}
\IEEEauthorblockA{\textit{University of Waterloo} \\
Waterloo, Canada\\
iang@uwaterloo.ca}
\and
\IEEEauthorblockN{Urs Hengartner}
\IEEEauthorblockA{\textit{University of Waterloo} \\
Waterloo, Canada\\
urs.hengartner@uwaterloo.ca}
\and
\IEEEauthorblockN{Florian Kerschbaum}
\IEEEauthorblockA{\textit{University of Waterloo} \\
Waterloo, Canada\\
florian.kerschbaum@uwaterloo.ca}
}

\maketitle
\begin{abstract}
    Training machine learning models on privacy-sensitive data has become a popular practice, driving innovation in ever-expanding fields. 
    This has opened the door to new attacks that can have serious privacy implications.
    One such attack, the Membership Inference Attack (MIA), exposes whether or not a particular data point was used to train a model.
    A growing body of literature uses Differentially Private (DP) training algorithms as a defence against such attacks.
    However, these works evaluate the defence under the restrictive assumption that all members of the training set, as well as non-members, are independent and identically distributed.
    This assumption does not hold for many real-world use cases in the literature.
    Motivated by this, we evaluate membership inference with statistical dependencies among samples and explain why DP does not provide meaningful protection (the privacy parameter $\epsilon$ scales with the training set size $n$) in this more general case.
    We conduct a series of empirical evaluations with off-the-shelf MIAs using training sets built from real-world data showing different types of dependencies among samples.
    Our results reveal that training set dependencies can severely increase the performance of MIAs, and therefore assuming that data samples are statistically independent can significantly underestimate the performance of MIAs.
\end{abstract}

\begin{IEEEkeywords}
Membership Inference Attacks, Differential Privacy
\end{IEEEkeywords}

\input{intro}
\input{preliminaries}

\input{theo}

\input{eval}
\input{discussion}

\input{related}

\input{conclusion}
\section*{Acknowledgements}
	We gratefully acknowledge the support of NSERC for grants RGPIN-05849, RGPIN-03858, IRC-537591, and the Postgraduate Scholarship-Doctoral/ CGS-M programs. We further acknowledge the Royal Bank of Canada, and the Waterloo-Huawei Joint Innovation Laboratory for funding this research.
    This research was undertaken, in part, thanks to funding
from the Canada Research Chairs program.
	This work was partially funded by an Ontario Graduate Scholarship (OGS).
	This work benefited from the use of the CrySP RIPPLE Facility at the University of Waterloo.

\bibliographystyle{IEEEtran}
\bibliography{refs}

\appendix
\input{appendix}

\end{document}

%% file: macros.tex
\usepackage{cite}
\usepackage[utf8]{inputenc} 
\usepackage[T1]{fontenc}    
\usepackage{hyperref}       
\usepackage{url}            
\usepackage{booktabs}       
\usepackage{amsfonts}       
\usepackage{nicefrac}       
\usepackage{microtype}      
\usepackage[usenames,svgnames,table]{xcolor}         
\usepackage{longtable}
\usepackage{xifthen}
\usepackage{commath}
\usepackage{amsmath ,amsthm}
\usepackage{thm-restate}
\usepackage{algorithm}
\usepackage[noend]{algpseudocode}
\usepackage{textcomp}
\usepackage{subfigure}
\usepackage{multirow}
\usepackage{float}
\usepackage{caption}
\usepackage{xurl}
\usepackage{wrapfig}
\usepackage{graphicx}
\usepackage{textcomp}
\usepackage{xcolor}
\newcommand{\ourtitle}{Investigating Membership Inference Attacks\\under Data Dependencies}
\newcounter{margin} 
\definecolor{scolor}{rgb}{1.0, 0.08, 0.58}
\definecolor{lcolor}{rgb}{0.0, 0.27, 0.13}
\definecolor{fcolor}{rgb}{0.1, 0.1, 0.44}
\definecolor{tcolor}{rgb}{1.0, 0.5, 0}
\definecolor{icolor}{rgb}{1.0, 0, 0}
\definecolor{ucolor}{rgb}{0, .5, 0}
\definecolor{bcolor}{rgb}{0.5, 0, 0.5}

\usepackage{pifont}
\newcommand{\TPR}{\text{TPR}}
\newcommand{\FPR}{\text{FPR}}

\newcommand{\ExpAltt}{\ensuremath{\texttt{Exp}_{\texttt{EX}}'}}
\newcommand{\ExpAlttt}{\ensuremath{\texttt{Exp}_{\texttt{EX}}''}}
\newcommand{\ExpIID}{\texttt{Exp}_{\texttt{IID}}}
\newcommand{\ExpEX}{\texttt{Exp}_{\texttt{EX}}}
\newcommand{\ExpNoIID}{\texttt{Exp}_{\texttt{MM}}}
\newcommand{\ExpSA}{\texttt{Exp}_{\texttt{STR}}}
\newcommand{\Att}{\texttt{Att}}

\newcommand{\Adv}{\texttt{Adv}}
\newcommand{\AdvAlt}{\texttt{Adv}_{\texttt{EX}}}
\newcommand{\AdvAltt}{\texttt{Adv}_{\texttt{EX}}'}

\newcommand{\dist}{\mathcal{D}}
\DeclareMathAlphabet\mathbfcal{OMS}{cmsy}{b}{n}
\newcommand{\distn}{\mathbfcal{D}}

\newcommand{\yeom}{threshold attack}
\newcommand{\optimal}{optimal threshold attack}
\newcommand{\shokri}{shadow model attack}
\newcommand{\average}{average threshold attack}
\newcommand{\genExp}{MM membership experiment}
\newcommand{\adult}{\texttt{adult}}
\newcommand{\compas}{\texttt{compas}}
\newcommand{\heart}{\texttt{heart}}
\newcommand{\texas}{\texttt{texas}}
\newcommand{\census}{\texttt{census}}
\newcommand{\students}{\texttt{students}}

\newtheorem{theorem}{Theorem}[section]

\theoremstyle{plain}
\theoremstyle{definition}
\newtheorem{defn}{Definition}[section]

%% file: intro.tex
\section{Introduction}
Machine learning (ML) is increasingly used to make predictions on privacy-sensitive data. 
In recent years, large tech companies such as Google and Amazon have begun to offer machine learning as a service to the general public through their cloud platforms. 
Although these systems can yield new and interesting insights in a variety of fields, machine learning models trained on sensitive data also present a lucrative attack surface for adversaries. 
In a Membership Inference Attack (MIA)~\cite{shokri2016membership, yeom2017privacy},
an adversary attempts to identify the data that was used to train an ML model (i.e., its members). 
Successful MIAs violate the privacy of individuals and pose a significant threat to unprotected ML models~\cite{shokri2016membership}.

Differentially private (DP) mechanisms can be applied to an ML model during the learning process \cite{chaudhuri_differentially_nodate,Abadi_2016,beaulieujones2018privacypreserving, shokri2016membership, yeom2017privacy, yu2019differentially} to limit the effect that a single data point can have on the model's output. 
Training ML models with DP mechanisms is becoming more common with popular libraries such as TensorFlow~\cite{tensorflow} and PyTorch~\cite{pytorch} offering several DP training algorithms. 

Previous work has provided experimental and theoretical evidence of DP training as a defence against MIAs~\cite{li2013membership,yeom2017privacy,jayaraman_evaluating_2019,shokri2016membership}.
Yeom et al.~\cite{yeom2017privacy} prove that DP training algorithms ensure a \emph{theoretical upper bound} on the privacy leakage caused by MIAs.
This bound, later improved by Erlingsson et al.~\cite{jayaraman2020revisiting}, offers strong protection for high DP levels (corresponding to small values of the DP parameter $\epsilon$), but quickly weakens for large $\epsilon$.
Despite this, empirical evaluations of MIAs against DP-trained models~\cite{jayaraman_evaluating_2019, jayaraman2020revisiting} suggest that 
large values of $\epsilon$, theoretically regarded as low-privacy settings, provide sufficient protection against MIAs in the wild.
Motivated by these findings, Murakonda and Shokri~\cite{murakonda2020ml} developed a tool for tuning $\epsilon$ based solely on the empirical performance of MIAs, thus increasing the model's utility at the expense of  theoretical privacy guarantees.
In short, this body of literature suggests that applying DP mechanisms to ML models, even in severely weakened privacy regimes, provides sufficient protection against known MIAs.

These works on membership inference all follow a similar evaluation methodology where members and non-members are sampled independently from the same distribution.
However, this underlying assumption is not always representative of real data~\cite{Mohammad2020, arjovsky2021distribution, pmlr-v97-recht19a, Martensson2020, nagarajan2021understanding, pmlr-v119-sagawa20a, sagawa2020distributionally}.
Furthermore, prior work has shown that data dependencies can lead to an overestimation of the privacy protection provided by DP~\cite{tschantz_sok_2020, kifer_no_2011, pufferfish, li2013membership, liu2016dependence}.
Motivated by this, we conduct the first study of MIA performance on ML models trained with DP-SGD~\cite{Abadi_2016} \emph{under data dependencies}.

We begin by studying the methodology from previous work on MIAs and identify the assumptions that underpin their theoretical findings~\cite{yeom2017privacy,jayaraman_evaluating_2019,nasr2021adversary}.
Since there are no publicly known instances of MIAs against ML models in the wild, researchers are limited to simulating attacks using theoretical games between the model owner and attacker called \emph{membership experiments}~\cite{yeom2017privacy}.
We show that previous membership experiments assume independence between training samples, which might not represent many membership inference cases in practice. 
To solve this, we propose a new membership experiment that loosens the restrictive independence assumption.
We demonstrate that previous DP-based bounds on the adversary's success do not hold in this more general setting.
We emphasize that this does not mean that DP is broken or flawed.
Rather, the guarantees of DP mechanisms no longer apply when defining membership inference with data dependencies.
Intuitively, this is because DP mechanisms limit the effect that a \emph{single} sample has on the output.
However, when training samples have dependencies, related samples can leak additional information about this single sample.

Prior work has shown the importance of studying MIA performance both theoretically and empirically~\cite{jayaraman_evaluating_2019, jayaraman2020revisiting}.
Thus, we also study the effect of data dependencies on the \emph{empirical performance} of MIAs under different instantiations of our membership experiment.
Using six real-world datasets, we evaluate the performance of off-the-shelf attacks.
These attacks have black-box query access to the target model and are given no additional background information about the data dependencies we consider.
Using these attacks, we investigate three different types of data dependencies.
First, to study the extent to which data dependencies can affect MIAs, we artificially split a dataset into members and non-members using a clustering algorithm.
Our results show that data dependencies can lead to perfect membership inference even in the presence of DP training.
Second, we investigate the effects of an explicit bias in an attribute of the training set.
In particular, we create a gender and an education bias in the training set, and see that MIA performance typically increases as the training set bias becomes more pronounced. 
Finally, we consider MIA performance on various real-world examples of data dependencies that occur during data collection.
For example, we study the effect of all members being from a specific health region (or hospital) and non-members from all other regions.
We also study the performance of MIAs when members and non-members come from the same source (the US Census) but are curated by different researchers.

In all of our evaluations with statistical dependencies among training samples we find a significant increase in MIA performance over the related work where all samples are Independent and Identically Distributed (IID)~\cite{shokri2016membership,jayaraman_evaluating_2019, jayaraman2020revisiting, yeom2017privacy,jagielski2020auditing, nasr2021adversary}.
This shows that the assumption of data independence among training samples underestimates the attack performance.
Our empirical evaluation yields attack performance greater than the bounds of DP, which confirms that the currently known DP bounds do not apply once we consider dependent data, and thus DP mechanisms are not a cure-all solution for membership inference.
This highlights the importance of considering data dependencies in future MIA evaluations and urges new research to develop better defences, and more realistic theoretical bounds on MIA performance.

%% file: preliminaries.tex
\section{Preliminaries}
In this section, we summarize the concepts related to membership inference attacks and differentially private machine learning that are most relevant to our work.
For reference, our notation is summarized in Table~\ref{table:notation}.

\begin{table}
	\centering
	\caption{Notation}\label{table:notation}
	\begin{tabular}{r p{6.5cm}}
		\textbf{Notation} & \textbf{Description}\\ \hline
		$z=(x,y)$ & Data point with feature vector $x$ and label $y$\\
		$\mathcal{Z}$ & Space of all possible data points\\
		$S$ & Training set of size $n$; $S\in\mathcal{Z}^n$\\
		$\dist$ & Distribution of a data point (over $\mathcal{Z}$)\\	
		$\distn$ & Joint distribution of $n$ data points (over $\mathcal{Z}^n$)\\
		$[K]$ & Set of integers from 1 to $K$\\
		$\mathcal{A}$ & Space of all possible machine learning models\\
		$A$ & Learning algorithm $A:\mathcal{Z}^n\to\mathcal{A}$\\
		$a$ & Instance of a trained model ($a\in\mathcal{A}$)\\ 
		\hline
		$\Att$ & Membership inference attack, outputs a bit\\ 
		$\Adv$ & Membership advantage, $\Adv\in[0,1]$\\
		\hline
	\end{tabular} 
\end{table}
We use $z=(x,y)$ to denote an element or data sample, where $x$ is its feature vector and $y$ is its class or label.
Let $\mathcal{Z}$ be the element space, i.e., $z\in\mathcal{Z}$.
We use $S\in\mathcal{Z}^n$ to denote a training set that contains $n$ elements $z\in\mathcal{Z}$.
Let $\dist$ be a probability distribution over $\mathcal{Z}$; $z\sim\dist$ means that $z$ is randomly sampled from $\dist$, and $S\sim\dist^n$ means that $S$ consists of $n$ independent samples from $\dist$.
We use $\distn$ to denote a probability distribution over $\mathcal{Z}^n$, i.e., a \emph{joint} distribution for all samples in a dataset; $S\sim\distn$ means the dataset $S$ is sampled from $\distn$.
We use $[K]$ to denote the set of integers from 1 to $K$.
A training algorithm $A$ is a (possibly randomized) function that takes a training set $S\in\mathcal{Z}^n$ and outputs a trained model $a\in\mathcal{A}$, where $\mathcal{A}$ is the space of trained models.
We use $a=A(S)$ to denote that $a$ is the model that results from applying the training algorithm $A$ to the training set $S$.
The goal of the trained model $a$ is to solve a classification task; i.e., assign a label $y$ to a feature vector $x$.
In this work, we focus on neural networks trained with DP-SGD~\cite{Abadi_2016}, which are a popular model choice for solving classification problems in machine learning.
However, our theoretical findings are generic and apply to other models as well.

\subsection{Membership Inference Attacks}
\label{sec:mia}
Though useful for solving classification tasks, machine learning models are subject to various privacy attacks. 
In this work we focus on the Membership Inference Attack (MIA), whose goal is to determine whether or not a specific data point $z$ was included in the training set of a target model $a$.
This attack is particularly dangerous when the model is trained on sensitive data, where an individual's inclusion or exclusion in the dataset could reveal sensitive or compromising information to an attacker. 
MIAs are related to property inference attacks~\cite{Ganju_2018,ateniese2013hacking,mahloujifar2022property}, but are fundamentally different.
In a Property Inference Attack (PIA), the goal of the attacker is to infer some property of the training set that the model producer did not intend to share (i.e., a property common to all training set samples), whereas the goal of an MIA is to infer whether or not a particular sample was in the training set. 
We discuss the relationship to PIAs further in Section~\ref{sec:disc}.

We use $\Att$ to denote a membership inference attack.
Typically, an MIA receives the target model $a$ and a data point $z$, and knows the training mechanism $A$ and some statistical information about the training data (e.g., $\dist$ or $\distn$).
The attack outputs a bit, which is $\Att=0$ when it decides that $z$ is a member, and $\Att=1$ when it decides it is a non-member.
Two of the most well-known membership inference attacks are the proposals by Shokri et al.~\cite{shokri2016membership} and Yeom et al.~\cite{yeom2017privacy}.
The \emph{\shokri}~by Shokri et al.~\cite{shokri2016membership} uses public data to train a set of shadow models, designed to mimic the target model's functionality.
Then, it trains an additional attack model to identify the membership status of a sample using outputs from the shadow models.
The \emph{\yeom}~by Yeom et al.~\cite{yeom2017privacy} assumes that the adversary has access to the loss function of the target model as well as the distribution of the loss on the private training data.
Given a data sample $z$ and the model $a$, the attack queries the model to obtain the loss of $z$, and decides that $z$ is a member if its loss is below a threshold computed from the distribution information.
A weaker variant of this attack assumes the adversary only knows the expected loss of the training set, and uses this value as decision threshold.

A variety of metrics can be used to measure the success of an MIA.
The True Positive Rate ($\TPR$) is the probability that the attack correctly identifies a member as such, and the False Positive Rate ($\FPR$) is the probability of incorrectly guessing that a non-member is a member.
A popular metric to measure the success of an MIA, which we use in this work, is the \emph{membership advantage}, which Yeom et al.~\cite{yeom2017privacy} define as $\Adv\doteq \TPR-\FPR$.
This metric is $0$ when the attack randomly decides the membership of $z$, and is $1$ when the attack always guesses the membership of $z$ correctly.

\subsection{Differential Privacy in Machine Learning}

Differential Privacy (DP), a privacy notion introduced by Dwork et al.~\cite{dwork_differential_2006}, has become the gold standard in database privacy and is a popular privacy notion in machine learning:

\begin{defn}[($\epsilon,\delta$)-DP]
\label{defn:dp}
 A training algorithm $A$ provides ($\epsilon,\delta$)-DP iff, for any two neighbouring datasets $S, S'\in \mathcal{Z}^n$ (i.e., $S$ and $S'$ differ by a single entry), and all possible subsets of the space of trained models $\mathcal{R}\subseteq\mathcal{A}$,
\begin{equation} \label{eq:dp}
 \Pr(A(S)\in\mathcal{R})\leq\Pr(A(S')\in\mathcal{R})\cdot e^\epsilon + \delta\,.
\end{equation}
\end{defn}
The parameter $\epsilon$ captures the degree of leakage of the mechanism $A$~\cite{dwork_differential_2006}.
Small values of $\epsilon$ indicate that a model trained with $S$ is indistinguishable from a model trained with $S'$ which, intuitively, makes it difficult to infer whether or not an element $z$ is in the training set.
The parameter $\delta$ makes it easier to satisfy the DP constraint by allowing a small chance of failure in the privacy guarantee.
It is typical to choose $\delta<1/n$, where $n$ is the number of elements in the dataset~\cite{dwork2014algorithmic}.

There are different approaches to developing a differentially private training algorithm $A$.
In the case of neural networks, the differentially private stochastic gradient descent technique by Abadi et al.~\cite{Abadi_2016} is widely used.
This technique involves clipping the gradients used for updating the network's weights during training time, and adding Gaussian noise to the average of the gradients.

%% file: theo.tex
\section{Membership Experiments}
\label{sec:theo}

Evaluating the performance of MIAs is crucial to understanding how dangerous these attacks are in practice and how to protect against them.
However, since there are no publicly known cases of MIAs against ML models in practice, researchers typically rely on \emph{membership experiments} to evaluate MIAs. 
A membership experiment is a theoretical game between the data owner and the adversary, where the adversary tries to guess the membership of a data sample.
The game specifies how the data owner generates the training data and which information is available to the attacker.
A membership experiment provides a \emph{theoretical definition of membership inference}, specifying the rules one has to follow when evaluating MIAs empirically, and allowing researchers to prove DP-based theoretical guarantees that hold within this controlled environment.

In this section, we review two membership experiments that appear in related work~\cite{yeom2017privacy,nasr2021adversary}.
We argue that these experiments make unrealistic assumptions and thus might not be representative of what one could expect in a real-world attack.
We then propose a new membership experiment that relaxes these assumptions.
We call this experiment the Mixture Model (MM) membership experiment.
We argue that the MM experiment better represents an MIA that could occur in practice, and thus it provides a more realistic definition of what membership inference is.

\subsection{Strong-Adversary Membership Experiment}

Differential privacy implies, by Definition~\ref{defn:dp}, that an adversary observing a model $a$ cannot easily distinguish whether it has been trained with $S=\tilde{S}\cup\{z\}$ or $S'=\tilde{S}\cup \{z'\}$, where $\tilde{S}\in\mathcal{Z}^{n-1}$, and $z,z'\in\mathcal{Z}$.
This naturally leads to the \emph{strong-adversary membership experiment} that we describe in Algorithm~\ref{alg:nasr} ($\ExpSA$).
This experiment has five inputs: an attack $\Att$, a training algorithm $A$, a set of $n-1$ samples $\tilde{S}$, and two data samples $z$ and $z'$.
The data owner flips a bit $b$ to decide whether to train a model with $S=\tilde{S}\cup\{z\}$ ($b=0$) or with $S'=\tilde{S}\cup\{z'\}$ ($b=1$).
The adversary receives the trained model $a$, the set $\tilde{S}$, and data samples $z$ and $z'$, and also knows the training algorithm $A$.
The adversary succeeds ($\ExpSA=1$) if it correctly learns whether $z$ or $z'$ was a member of the training set of $a$ (i.e., bit $b$).
Nasr et al.~\cite{nasr2021adversary} use this membership experiment, formulated as a privacy game, to prove that the $\epsilon$-DP guarantees of current training algorithms are tight.\footnote{Nasr et al.~\cite{nasr2021adversary} use the so-called \emph{unbounded} DP notion, where one of the neighbouring datasets has one additional element (i.e., $S$ vs.~$S'=S\cup\{z\}$). We stick to the so-called \emph{bounded} DP notion ($S$ and $S'$ have the same size) here since this is the privacy notion followed by Yeom et al. in their membership experiment~\cite{yeom2017privacy} }

\begin{algorithm}
	\begin{algorithmic}[1]
	\Procedure{$\ExpSA$}{$\Att, A, \tilde{S}, z, z'$}
	\State Choose $b\sim\{0,1\}$ uniformly at random;
	\If{$b=0$}
		\State Train $a=A(\tilde{S}\cup\{z\})$;
	\Else
		\State Train $a=A(\tilde{S}\cup\{z'\})$;
	\EndIf
	\State Return 1 if $\Att(a, z, z', \tilde{S}, A)=b$; else 0.	
	\EndProcedure
	\end{algorithmic}
	\caption{Strong-Adversary Membership Experiment~\cite{nasr2021adversary}}
	\label{alg:nasr}
\end{algorithm}

The membership advantage in $\ExpSA$ is
\begin{align*}
 \Adv(\Att, A, \tilde{S}, z, z')=\Pr(\Att=0|b=0)-\Pr(\Att=0|b=1)\,.
\end{align*}
Here, the positive summand is the $\TPR$ and the negative one is the $\FPR$ (if we consider that $b=0$ is a positive).
Note that these terms also depend on the input variables ($\Att, A, \tilde{S}, z, z'$) but we have omitted this dependence in the notation for simplicity.
Next, we prove that an $(\epsilon,\delta)$-DP training algorithm provides the following upper bound on the membership advantage:
\begin{theorem}[New Membership Advantage Bound]
\label{theo:new}
Let $A$ be an $(\epsilon,\delta)$-DP training algorithm. 
Then, for all attacks $\Att$, sets $\tilde{S}$, and data samples $z$ and $z'$, the membership advantage in $\ExpSA$ satisfies
\begin{equation} \label{eq:bound_strongadv}
 \Adv(\Att, A, \tilde{S}, z, z')\leq (e^\epsilon-1+2\delta)/(e^\epsilon+1)\,.
\end{equation}
\end{theorem}

\begin{proof}

Kairouz et al.~\cite{kairouz2017composition} show that, for an adversary that wants to distinguish between two neighbouring inputs of a DP mechanism, the following bounds on the adversary's TPR and FPR hold: 
\begin{equation*}
	\FPR+e^\epsilon\cdot(1-\TPR)\geq 1-\delta\,,\quad
	(1-\TPR)+e^\epsilon \cdot \FPR\geq 1-\delta\,.
\end{equation*}
These expressions follow from \eqref{eq:dp} and trivially apply to the TPR and FPR in $\ExpSA$.
Adding these expressions together we get
\begin{equation*}
	(1-\TPR+\FPR)(1+e^\epsilon)\geq 2(1-\delta)\,.
\end{equation*}
Applying the definition of membership advantage ($\Adv\doteq\TPR-\FPR$) and leaving $\Adv$ in one side of the inequality yields the desired upper bound.
\end{proof}

For values of $\delta$ and $\epsilon$ close to zero, the bound above is close to 0.
This means that DP is potentially a very strong privacy guarantee since it ensures that the advantage of an adversary that \emph{knows all training set samples except one} (i.e., $\tilde{S}$), and has two possible candidates for the unknown sample ($z$ and $z'$) will always be below a certain value.
However, as Nasr et al.~\cite{nasr2021adversary} admit, this adversary is often unrealistically strong, and perhaps assessing the performance of MIAs that are given all training set samples but one is not very helpful towards understanding MIA in practical cases.

Another problem with this experiment is the \emph{relation} between the membership bit $b$, the observed samples $z$ and $z'$, and the trained model $a$.
In $\ExpSA$, the bit $b$ determines the training set of $a$ and thus they are causally \emph{related} ($\Pr(a|b)\neq\Pr(a)$).
We argue that there is no such dependence in an actual MIA scenario.
In practice, the data owner trains their model once, with their (one and only) training set $S$.
Then, the adversary receives the model and tries to determine whether a sample $z$ is a member of its training set.
Here, the bit $b$ is \emph{a property of $z$} representing its membership status, and not a property of the model $a$, i.e., $\Pr(a|b)=\Pr(a)$. 
Yeom et al.'s experiment~\cite{yeom2017privacy}, described next, captures this relation.

\subsection{IID Membership Experiment}

Yeom et al.~\cite{yeom2017privacy} propose the membership experiment that we show in Algorithm~\ref{alg:yeom}
This experiment does not receive the training set $S$ as an input, but a distribution $\dist$ from which the samples of $S$ are independently sampled (IID).
The data owner trains $a$ using $S$ and then chooses a bit $b\in\{0,1\}$ uniformly at random.
This bit determines whether the adversary receives a member ($z$ chosen randomly from $S$) or a non-member (a fresh sample $z\sim\dist$).
The adversary receives the element $z$ and the trained model $a$, and knows the training set size $n$, the training algorithm $A$, and the element distribution $\dist$.
With this information, the adversary carries out an attack $\Att(z, a, n, A, \dist)$ that outputs a bit, indicating the predicted membership of $z$.
The adversary succeeds (denoted $\ExpIID=1$) if the attack correctly infers the membership status of $z$.

\begin{algorithm}
	\begin{algorithmic}[1]
	\Procedure{$\ExpIID$}{$\Att, A, n, \dist$}
	\State Sample $S\sim \dist^n$; \label{line:sampleyeom}
	\State Train $a=A(S)$;
	\State Choose $b\sim\{0,1\}$ uniformly at random;
	\If{$b=0$}
		\State Draw $z\sim S$;
	\Else
		\State Draw $z\sim \dist$;
	\EndIf	
	\State Return 1 if $\Att(z, a, n, A, \dist)=b$; else 0.
	\EndProcedure
	\end{algorithmic}
	\caption{IID Membership Experiment~\cite{yeom2017privacy}}
	\label{alg:yeom}
\end{algorithm}

Yeom et al.~\cite{yeom2017privacy} prove the following upper bound on the membership advantage in $\ExpIID$:

\begin{theorem}[Yeom et al.'s bound~\cite{yeom2017privacy}] 
\label{theo:yeom_bound}
	Let $A$ be an ($\epsilon,0$)-DP learning algorithm. Then, for all attacks $\Att$, training set sizes $n$, and data point distributions $\dist$, the membership advantage in $\ExpIID$ satisfies
	\begin{equation} \label{eq:yeombound}
	\Adv(\Att,A,n,\dist)\leq e^\epsilon - 1\,.
	\end{equation}
\end{theorem}
Note that the advantage is also upper bounded by $1$, so this bound is loose for $\epsilon>\ln 2$.
More recently, Erlingsson et al.~\cite{erlingsson_that_2020}
provide a tighter bound using results from Hall et al.~\cite{hall2013differential} for the more generic setting of ($\epsilon,\delta$)-DP:
\begin{theorem}[Erlingsson et al.'s bound~\cite{erlingsson_that_2020}] 
\label{theo:google_bound}
	Let $A$ be an ($\epsilon$,$\delta$)-DP learning algorithm. Then, for all attacks $\Att$, training set sizes $n$, and data point distributions $\dist$, the membership advantage in $\ExpIID$ satisfies
	\begin{equation} \label{eq:googlebound}
	\Adv(\Att,A,n,\dist)\leq 1 - e^{-\epsilon} ( 1 - \delta)\,.
	\end{equation}
\end{theorem}
We can see that this bound is less than or equal to 1.

In Appendix~\ref{sec:app}, we prove that the advantage bound we derive in Theorem~\ref{theo:new} holds in the case where members and non-members are \emph{statistically exchangeable}, i.e., when the joint distribution of a sequence of $n$ members and one non-member does not depend on the order they appear in the sequence.
We also explain the IID experiment is a particular case of statistical exchangeability, and thus our bound holds in this case:
\begin{restatable}{theorem}{newIID}[New Membership Advantage Bound]
\label{theo:newIID}
Let $A$ be an ($\epsilon$,$\delta$)-DP learning algorithm. 
Then, for all attacks $\Att$, training set sizes $n$, and data point distributions $\dist$, the membership advantage in $\ExpIID$ satisfies
\begin{equation} \label{eq:newbound}
 \Adv(\Att,A,n,\dist)\leq (e^\epsilon-1+2\delta)/(e^\epsilon+1)\,.
\end{equation}
\end{restatable}
This bound improves over previous results~\cite{yeom2017privacy,erlingsson_that_2020}; i.e., it is tighter than \eqref{eq:yeombound} and \eqref{eq:googlebound} for all $\epsilon\geq 0$ and $0\leq\delta\leq 1$.
We prove this in Appendix~\ref{app:tightest}.
Figure~\ref{fig:bounds} shows a comparison between the existing membership advantage upper bounds (Yeom et al.~\eqref{eq:yeombound} and Erlingsson et al.~\eqref{eq:googlebound}) and our bound \eqref{eq:newbound}.
We used a value of $\delta=10^{-5}$ for \eqref{eq:googlebound} and \eqref{eq:newbound}, since this is the value of $\delta$ we use in our experiments (described in Section~\ref{sec:setup}) to ensure that $\delta$ is less than the inverse of the training set size.
Our bound improves Erlingsson et al.'s by almost $0.2$ in advantage for relevant privacy values $0.1\leq\epsilon\leq 2$.
These bounds suggest that high privacy settings $0.01\leq \epsilon\leq 0.1$ completely thwart membership inference attacks (in the IID setting),
values of $\epsilon\approx 1$ achieve intermediate privacy levels, and large values of $\epsilon\geq 10$ provide no worst-case privacy guarantee on the membership advantage.
\begin{figure}
	\centering
		\includegraphics[width=0.9\linewidth]{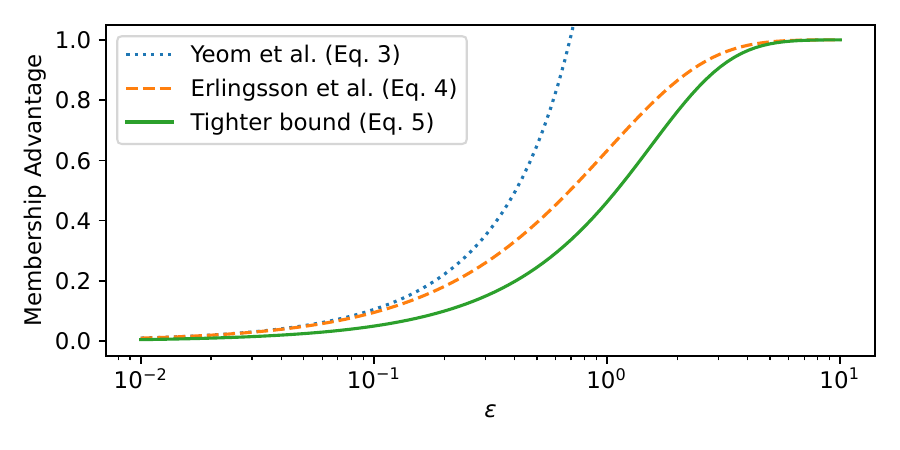}
	\caption{Upper bounds on the membership advantage in the IID Membership Experiment.}\label{fig:bounds}
\end{figure}

We note that we have considered the original experiment by Yeom et al.~\cite{yeom2017privacy}, where members and non-members are sampled independently from $\dist$.
In practice, members and non-members are typically different samples, so \emph{sampling without replacement} is a more realistic setting.
Sampling without replacement introduces dependencies among members and non-members, but the exchangeability property still holds in that case, and thus our bound still applies.
For simplicity, in the remainder of the paper we use the term \emph{IID sampling} to denote either \emph{sampling with or without replacement}, since our findings apply to both settings.
The IID membership experiment captures the evaluation approach of most related works on membership inference~\cite{shokri2016membership,jayaraman_evaluating_2019, jayaraman2020revisiting, yeom2017privacy,jagielski2020auditing, nasr2021adversary}.
These works build the member and non-member sets by randomly splitting a dataset, which simulates IID sampling.
However, this procedure excludes the case where there are dependencies between training set samples (e.g., a training set bias).
There is a large body of literature in ML that studies this scenario with the goal of improving the accuracy of models on out-of-distribution samples~\cite{Mohammad2020, arjovsky2021distribution, pmlr-v97-recht19a}.
For example, in a health care setting such as MRI scanning, the images seen in the wild are often from a different distribution to the training data of the model~\cite{Martensson2020}.
Spurious correlations in ML datasets have also received a lot of attention in the ML literature as they can have a detrimental effect on model accuracy~\cite{nagarajan2021understanding, pmlr-v119-sagawa20a, sagawa2020distributionally}.
Despite the problems of non-IID data being known to the ML community, this case has been completely neglected in the MIA literature.

\subsection{A Membership Experiment with Data Dependencies}
\label{sec:gen}

We have argued that the strong-adversary membership experiment considers an unrealistically strong adversary and does not capture the idea that membership should be a property of a given data point, and not the model itself.
Also, the IID experiment is limited to a scenario that, as suggested by related work in ML, does not hold in many use cases.
Motivated by this, we propose a new membership experiment that extends $\ExpIID$ to cases where the training set samples may have dependencies. 
We consider the particular case where the joint dataset distribution, denoted $\distn$, is a \emph{multivariate mixture model}.
A mixture model is a probabilistic model that characterizes the distribution of an overall population by taking into account the presence of subpopulations within the overall population.
Each of the subpopulations is characterized by a probability distribution called a \emph{mixture component}.
To draw one sample from a mixture model, one first selects one of the mixture components (subpopulations) at random, and then draws a sample from the mixture component.
Note that the mixture components may have overlapping support, which implies that, given a point sampled from the mixture model (i.e., the overall population), it might not be possible to identify from which of the mixture components (i.e., subpopulations) it was sampled from.
We use \emph{multivariate} mixture components, i.e., each component models the distribution of all attributes of a data sample within a subpopulation.
This particular case of $\distn$ is amenable to theoretical analyses (e.g., Sect.~\ref{sec:boundsdonothold}) and is general enough to accommodate many real-life examples, as we show in Section~\ref{sec:dataset}.
We call this the Mixture Model (MM) membership experiment, and denote it by $\ExpNoIID$.

\begin{algorithm}
	\begin{algorithmic}[1]
	\Procedure{$\ExpNoIID$}{$\Att, A, n, \distn$}
	\Comment{$\distn=\{\dist_1, \dots, \dist_K\}$}
	\State Choose $k\sim[K]$, sample $S\sim\dist^n_k$; \label{line:samp1}
	\State Train $a=A(S)$;
	\State Choose $b\sim\{0,1\}$ uniformly at random;
	\If{$b=0$}
		\State Draw $z\sim S$;
	\Else
		\State Choose $k'\sim[K]\setminus k$, draw $z\sim\dist_{k'}$; \label{line:samp2}
	\EndIf
	\State Return 1 if $\Att(z, a, n, A, \distn)=b$; else 0.
	\EndProcedure
	\end{algorithmic}
	\caption{Mixture Model (MM) Membership Experiment}
	\label{alg:ExpJoint}
\end{algorithm}

Algorithm~\ref{alg:ExpJoint} shows the MM membership experiment.
The joint distribution $\distn$ is a multivariate mixture model with $K$ mixture components.
Each mixture component $k\in[K]$ models the training set samples independently following a distribution $\dist_k$.
We overload the notation of $\distn$ in this case for notational simplicity, and use it to denote the set of all mixture distributions $\distn=\{\dist_1,\dots,\dist_K\}$.
In short, to sample $S\sim\distn$, we choose a mixture component uniformly at random ($k\sim[K]$) and then draw $n$ samples \emph{independently} from $\dist_k$ (line~\ref{line:samp1}).
Note that, even though the samples are independently drawn from a mixture component $\dist_k$, they are statistically dependent since they are drawn from the \emph{same} (unknown) mixture component.
When $b=1$, the non-member $z$ is sampled from a mixture component that is \emph{different} from the training set's component; i.e., $z\sim\dist_{k'}$, for $k'\neq k$ (line~\ref{line:samp2}). 
This experiment can also encompass the IID case in $\ExpIID$ when all $\dist_k$ (for $k\in[K]$) are identical.

The \genExp~accommodates many realistic settings.
As an example, consider a case where there are $K$ hospitals in a city.
Let $\dist_k$ be the distribution that models the samples from the $k$th hospital, for all $k\in[K]$.
In this example, $\ExpNoIID$ represents the case where the target model $a$ is trained with samples belonging to a particular hospital $k\in[K]$, and the samples from all the other hospitals $k'\in[K]\setminus k$ are non-members.
If the distributions of samples in each hospital are slightly different, there will be certain dependencies among members and non-members.
Previous theoretical experiments ($\ExpSA$ and $\ExpIID$) cannot account for these dependencies.
In Section~\ref{sec:eval} we study this and other realistic examples that follow $\ExpNoIID$.

We note that, following $\ExpNoIID$, the adversary has statistical knowledge about the possible dependencies between training set samples ($\distn$).
However, since the choice of $k$ is not revealed, the adversary cannot tell which mixture component distribution the members follow.
Even though this statistical knowledge is powerful, it does not provide the adversary any a-priori membership advantage.
Indeed, note that, when $\epsilon=0$ (or, alternatively, if the adversary did not receive the target model $a$), the adversary cannot do better than randomly guessing the membership of $z$ ($\Adv=0$).
This is because, without the model $a$, the distribution of $z$ is the same regardless of the membership bit $b$.
However, when $\epsilon>0$, the adversary can exploit the statistical dependencies among samples to improve the membership inference.
We also note that, in our evaluation in Section~\ref{sec:eval}, we evaluate off-the-shelf MIAs that \emph{do not use the statistical dependencies} captured by $\distn$.
However, we will see that these dependencies still cause these existing attacks to perform better than in the IID scenario.

The \genExp~captures the relation between $b$, $z$, and $a$, considers a realistic adversary that only has statistical knowledge about the dataset, and accommodates both IID training set distributions ($\dist_k\equiv\dist_{k'}$, $\forall k, k'\in[K]$) as well as those that incorporate dependencies among samples or bias ($\dist_k\not\equiv\dist_{k'}$).
Therefore, we believe that this experiment provides a better template to define membership inference, and that evaluating MIAs under this experiment will provide results that are more representative than previous experiments of what one could expect in practice.

\subsubsection{DP-based bounds on the generalized membership experiment}
\label{sec:boundsdonothold}
We note that existing DP-based bounds~\cite{yeom2017privacy,erlingsson_that_2020} on the membership advantage $\Adv$, as well as our new bound in Eq.~\eqref{eq:newbound}, do not apply in the \genExp.
We note that this is not a flaw of DP: the training algorithms we evaluated provide differential privacy, but this is not enough to bound the leakage when defining membership inference as in $\ExpNoIID$. 
The reason for this is that DP mechanisms limit the effect that a \emph{single} sample has on the output (see Def.~\ref{defn:dp}).
However, when training samples have dependencies, limiting the effect of a single input is less effective since the remaining data points can also leak information about this input due to their co-dependencies~\cite{liu2016dependence}.

We provide more intuition as to why the previous bounds do not hold with an example.
Without loss of generality and for simplicity, in this analysis we assume that the space of samples $\mathcal{Z}$ is discrete.
In both $\ExpIID$ and $\ExpNoIID$, the adversary observes the released model $a$ and a sample $z$ (along with general parameters $n$, $A$, $\dist$, that we omit from the probability expressions here for notational simplicity).
The joint distribution of $a$ and $z$ depends on whether $z$ was a member of $a$'s training set ($b=0$), or a non-member ($b=1$).
Next, we write general expressions for the likelihood of $a$ and $z$ given the value of $b$.
We use these expressions to compute the maximum likelihood ratio, which bounds the adversary advantage, and study how an $(\epsilon,0)$-DP mechanism bounds this ratio in the IID and non-IID cases.

First, when $z$ is a member ($b=0$), and using $\tilde{S}$ to denote all other $n-1$ members, we can write
\begin{equation*}
	\Pr(a,z|b=0)=\Pr(z)\cdot\sum_{\tilde{S}} \Pr(\tilde{S}|z)\Pr(A(\tilde{S}\cup \{z\})=a)\,.
\end{equation*}
Here, $\Pr(\tilde{S}|z)$ is the probability of $n-1$ training set samples (all but $z$) conditioned on the fact that $z$ is also a training set sample.

When $z$ is a non-member (denoted $b=1$), $\tilde{S}$ still denotes $n-1$ members, and we use $z'$ to define the $n$th member.
Then, using the law of total probability we can write the likelihood as
\begin{equation*}
    \resizebox{\hsize}{!}{
	$\displaystyle\Pr(a,z|b=1)=\Pr(z)\sum_{z'} \Pr(z')\sum_{\tilde{S}} \Pr(\tilde{S}|z')\Pr(A(\tilde{S}\cup \{z'\})=a).$}
\end{equation*}
(As above, $\Pr(\tilde{S}|z')$ is the probability of the $n-1$ training set samples $\tilde{S}$ given the remaining training sample $z'$.)

The maximum likelihood ratio $L_{max}=\max_{a,z,\mathfrak{b}}\frac{\Pr(a,z|b=\mathfrak{b})}{\Pr(a,z|b=1-\mathfrak{b})}$ can be used to bound the adversary advantage.\footnote{We omit the derivations for space issues; in short, $L_{max}$ lower-bounds the adversary's probability of error by $1/(1+L_{max})$, which in turn upper-bounds the advantage by $\Adv\leq(L_{max}-1)/(L_{max}+1)$.
Note this matches \eqref{eq:newbound} for $\delta=0$.}
We show why using an $(\epsilon,0)$-DP training algorithm provides a strong bound on this ratio in the IID case ($\ExpIID$), but there are non-IID cases ($\ExpNoIID$) where the bound is not meaningful.
First, recall that an $(\epsilon,0)$-DP training algorithm ensures that, for two training sets $S$ and $S'$ that differ in \emph{one} sample, and for all models $a\in\mathcal{A}$, $\Pr(A(S)=a)\leq e^\epsilon \cdot \Pr(A(S')=a)$.
In the IID scenario ($\ExpIID$), samples are independent and therefore $\Pr(\tilde{S}|z)=\Pr(\tilde{S})$.
The likelihood ratio between hypotheses $b=1$ and $b=0$ is bounded by:
\begin{align*}
  \frac{\Pr(a,z|b=1)}{\Pr(a,z|b=0)}&=\frac{\sum\limits_{z'} \Pr(z')\sum\limits_{\tilde{S}} \Pr(\tilde{S})\Pr(A(\tilde{S}\cup \{z'\})=a)}{\sum\limits_{\tilde{S}} \Pr(\tilde{S})\Pr(A(\tilde{S}\cup \{z\})=a)}\\
		&\leq\frac{\sum\limits_{z'} \Pr(z')\sum\limits_{\tilde{S}} \Pr(\tilde{S})\Pr(A(\tilde{S}\cup \{z\})=a) e^\epsilon}{\sum\limits_{\tilde{S}} \Pr(\tilde{S})\Pr(A(\tilde{S}\cup \{z\})=a)}\\
		&=\sum\limits_{z'}\Pr(z')\cdot e^\epsilon = e^\epsilon\,.
\end{align*}
The same applies to the reciprocal $\Pr(a,z|b=0)/\Pr(a,z|b=1)$. 
In the \genExp, if we are explicitly in a non-IID case, then $\Pr(\tilde{S}|z)\neq\Pr(\tilde{S})$ and the derivations above cannot be carried out.
As an example, consider the pathological case where the mixture components in $\distn$ are different and deterministic, i.e., all the training set samples are identical.
In that case, $\Pr(\tilde{S}|z)$ is only non-zero when $\tilde{S}=\{z\}^{n-1}$ (in that case, $\Pr(\{z\}^{n-1}|z)=1$).
Then,
\begin{align*}
  \frac{\Pr(a,z|b=1)}{\Pr(a,z|b=0)}&=\frac{\sum\limits_{z'} \Pr(z')\sum\limits_{\tilde{S}} \Pr(\tilde{S}|z')\Pr(A(\tilde{S}\cup \{z'\})=a)}{\sum\limits_{\tilde{S}} \Pr(\tilde{S}|z)\Pr(A(\tilde{S}\cup \{z\})=a)}\\
		&=\sum\limits_{z'} \Pr(z')\frac{\Pr(A(\{z'\}^{n})=a)}{\Pr(A(\{z\}^{n})=a)}\,.
\end{align*}
We can apply the DP definition $n$ times and bound the expression above by $e^{n\cdot\epsilon}$, but this is certainly not useful since $n$ is typically very large.
More generally, bounding the expression above requires using DP to provide \emph{group privacy}, which does not yield useful bounds for large $n$~\cite{dwork2014algorithmic}. 
Even though this is an extreme example, it illustrates why the previous DP-based bounds (that are independent of dataset size) on the membership advantage do not apply to $\ExpNoIID$.

The pathological example above shows that an upper bound on $\Adv$ can be very large in non-IID scenarios.
However, there might be particular distributions $\distn$ for which one could prove an upper bound that lies closer to the IID bound~\eqref{eq:newbound}.
Also, empirical MIAs have been shown to achieve advantage levels $\Adv$ far below DP-based upper bounds~\cite{jayaraman_evaluating_2019} in the IID scenario.
Thus, it is important to empirically measure the performance of MIAs in the non-IID case, as performance could be similar to the IID case.
We study this in the next section.

%% file: eval.tex
\section{Empirical Evaluation of Membership Inference with Data Dependencies}
\label{sec:eval}

In this section, we empirically investigate the effects of non-IID training sets on the performance of off-the-shelf MIAs, following our \genExp~$\ExpNoIID$.

\begin{table*}[t]
	\centering
	\caption{Evaluation Setup Summary}\label{tab:setup}
	\begin{tabular}{c | c |c c | c c | c }
		Evaluation & \multirow{2}{*}{Dataset and experiment} & Number of Features 	&  \multirow{2}{*}{Number of classes} & \multirow{2}{*}{Members} 	& \multirow{2}{*}{Non-members}  	&
		\multirow{2}{*}{Shadow model data} \\  
		(Section)  & 																		& (after encoding)		& 																		&														& 								 								&
		 \\ \hline
\multirow{2}{*}{\ref{sec:cluster}} 	& \adult~cluster       	& $104$ & $2$   & $10\,000$ & $10\,000$ & $15\,132$ \\
																		& \compas~cluster      	& $15$  & $2$   & $2\,000$  & $2\,000$  & -         \\ \hline
\multirow{2}{*}{\ref{sec:att}} 			& \adult~education     	& $87$  & $2$   & $10\,000$ & $10\,000$ & $10\,772$ \\
																		& \adult~gender        	& $103$ & $2$   & $10\,000$ & $10\,000$ & $11\,192$ \\ \hline
\multirow{7}{*}{\ref{sec:dataset}}	& \heart,~Cleveland		 	& $22$  & $2$   & $303$     & $617$     & -         \\
																		& \heart,~Hungary		 		& $22$  & $2$   & $294$     & $626$     & -         \\
																		& \heart,~Switzerland 	& $22$  & $2$   & $123$     & $797$     & -         \\
																		& \heart,~VA Long Beach	& $22$  & $2$   & $200$     & $720$     & -         \\
																		& \students,~Gabriel Pereira& $42$  & $2$   & $423$     & $226$     & -         \\
																		&\texas, region $\#3$ & $269$ & $100$ & $10\,000$ & $10\,000$ & $20\,000$ \\
																		& \census~and \adult~   & $100$ & $2$   & $10\,000$ & $10\,000$ & $20\,000$ \\
	\end{tabular}
\end{table*}

\subsection{Evaluation Setup}\label{sec:setup}
We explain how we instantiate $\ExpNoIID$ with real-world datasets.
In each of our experiments, we initialize $K$ disjoint sets of samples $\{D_k\}_{k\in[K]}$ with real data (in most of our experiments, we use $K=2$).
Each of these sets characterizes one of the mixture components: drawing a sample from $\dist_k$ is simply drawing a sample from the set $D_k$.
We use sampling \emph{without replacement} to ensure that every member and non-member is a distinct data sample. 
Even though this is an obvious assumption in practice (duplicate samples do not make sense in most problems), we note that all the theoretical experiments in Section~\ref{sec:theo} allow for identical members and non-members.

In each of the sections below, we follow a different approach to build the sets of samples $\{D_k\}_{k\in[K]}$, to study different types of data dependencies.
In each experiment, given a training set size $n$, we take $n$ samples from a mixture component $D_k$ as the member set, and draw $n$ samples at random from all other mixture components as non-members.
We train the model on the members, and evaluate each of the attacks we consider on all members and non-members.
We compute an attack's $\TPR$ by looking at the proportion of members correctly identified as such, and the $\FPR$ by computing the proportion of non-members wrongly classified as members by the attack.
Then, the attack's membership advantage is simply $\Adv=\TPR-\FPR$.
We label this advantage as ``non-IID'' in our plots.
We also measure the classification accuracy of the model over members and non-members in this non-IID case, to give insight into the attacks' performance.

In order to study the effect of data dependencies vs.~independent data sampling, we also measure the membership advantage if members and non-members were IID samples.
To do this, for each non-IID experiment, we merge the member and non-member sets, and generate a new member and non-member set by resampling from this merged set (we keep the original member and non-member set sizes).
The advantage in this case represents the attack's performance in the IID experiment, and we label it ``IID'' in our plots.

\paragraph{Datasets}

We use three publicly available datasets considered in prior work~\cite{kulynych2022disparate, shokri2016membership} and three additional datasets from the UCI machine learning repository~\cite{UCI_Repo}:
\begin{enumerate}
    \item The \adult~dataset~\cite{adult_dataset}, which contains $48\,842$ data samples from the 1994 US census, each with 14 attributes such as age, gender, and education. The binary classification task is deciding whether the individual earns more than \$50k a year.
    \item The \compas~dataset~\cite{compas_dataset}, extracted from ProPublica's investigation into racial bias in ML, which contains $6\,172$ samples with 15 attributes such as age, sex, and number of prior offenses.
		The classification task is predicting whether or not an individual re-offended within 2 years.
    \item The \texas~dataset~\cite{texas_dataset} contains a series of records of inpatient stays in various hospitals published by the Texas Department of State Health Services. The classification task is predicting the procedure. 
		We follow the approach by Shokri et al.~\cite{shokri2016membership} and compute the top-$100$ most popular procedures as the classification label (we discard any rows not in the top-$100$). 
		The final dataset contains $350\,280$ samples each with $66$ attributes such as length of stay, age, and total charges.
    \item The \heart~datasets~\cite{heart_dataset} are four datasets collected from hospitals in Cleveland, Hungary, Switzerland, and the VA Long Beach. 
		The datasets contain $303$, $294$, $123$, and $200$ data samples, respectively. 
		We use the $14$ attributes common to all datasets such as age, cholesterol, and fasting blood sugar with the classification task of predicting if the patient has a heart condition.
    \item The \students~dataset~\cite{student_dataset} contains data samples from a study on student achievement in Portuguese schools. 
		We use the Portuguese language dataset, which has $649$ samples from two schools ($423$ and $226$ samples each) containing $30$ attributes such as age, study time, and number of absences.
		The classification task is predicting if the student passed or failed the course, based on the final grade.
		We remove the two intermediate grades from the attribute list, as they are highly correlated with the final grade.
    \item The \census~dataset~\cite{census_dataset} contains data extracted from the 1994 and 1995 US Census. 
		We take all $99\,827$ data samples from the 1994 dataset and extract the $10$ attributes in common with the \adult~dataset. 
		The classification task is to predict whether an individual earns more than \$50k a year.
\end{enumerate}
We replace missing values with the mean or mode value for that attribute and use a one-hot encoding for all categorical attributes.
For example, if the categorical attribute has values red, green, and blue we add three new binary attributes, one for each colour.
Since we only used a subset of the attributes for certain datasets, the rows were not always unique. 
To address this, we remove duplicate rows, keeping a single copy of each duplicate.
Table~\ref{tab:setup} summarizes how we use these datasets in our experiments.
In most experiments, we set a value for $n$ ($n=10\,000$ in most cases, and $n=2\,000$ in \compas), and sample $n$ members and $n$ non-members using the procedure described above.
When the remaining data is large enough, we use it to evaluate the \shokri~(when we do not have enough data, denoted ``--'' in Table~\ref{tab:setup}, we skip this attack).
In experiments where the datasets are small (\heart~and \students), we use an uneven number of members and non-members (we explain this in Section~\ref{sec:dataset}).

\paragraph{Model Architecture}
Following the work of Jayaraman et al.~\cite{jayaraman_evaluating_2019}, and others~\cite{shokri2016membership, Abadi_2016}, we use a ReLU network with 2 hidden layers, each with 256 neurons trained for 100 epochs as the target model in all evaluations.
The models use the DP ADAM Gaussian Optimizer and RDP accountant from TensorFlow Privacy~\cite{tf_privacy} with $\ell_2$ regularization to avoid overfitting.
The default hyper-parameters are $10^{-5}$ as the  $\ell_2$ regularization coefficient, $10^{-2}$ as the learning rate, and $200$ as the batch size~\cite{jayaraman_evaluating_2019}.
We vary $\epsilon$ in our evaluations and fix $\delta$ to equal $10^{-5}$. 
This follows the general recommendation that $\delta$ should be less than the inverse of the training set size~\cite{dwork2014algorithmic}.
We note that $\epsilon=0$ still leaks information, since $\delta>0$.

\paragraph{Membership Inference Attacks}
We consider two membership inference attacks evaluated in prior work~\cite{jayaraman_evaluating_2019} that we summarized in Section~\ref{sec:mia}.
The \emph{\shokri}~\cite{shokri2016membership} requires a large pool of publicly available data in order to train shadow models. 
Therefore we only run this attack on datasets that have leftover samples after building our member and non-member sets (see Table~\ref{tab:setup}).
Following prior work~\cite{jayaraman_evaluating_2019, shokri2016membership}, we train five shadow models, each with the same architecture as the target model.
The attack model uses the same architecture as the target model but with $64$ neurons in the hidden layer.
We consider two variants of the \emph{\yeom}~\cite{yeom2017privacy}: one that knows the true loss distribution and therefore chooses the optimal decision threshold (\optimal) and one that uses the average training loss as decision threshold (\average).
We remark that we do not modify these attacks.
The attacks only get black box query access to the model $a$ (with confidence scores), the training loss (log loss attack), public data following $\mathcal{D}$ (shadow model attack), and the sample $z$. That is, they do not use all the information that we allow the attacker in Algorithm~\ref{alg:ExpJoint}. Specifically, we omit the set size $n$ and the training algorithm $A$.
The attacks were not designed to exploit dependencies between training samples and the adversary never explicitly learns these dependencies.
Furthermore, in each experiment, we are careful not to include attributes that can trivially identify the mixture component that a sample was drawn from (we clarify this below).

\paragraph{Implementation}
We used Python 3.8 for our implementation\footnote{Source code available at \url{https://github.com/t3humphries/non-IID-MIA}}, building upon the code of Jayaraman et al.~\cite{jayaramen_code}.
This code includes the implementation of the \yeom~and \shokri, where the latter is based on Shokri et al.~\cite{shokri2016membership}. 
To implement RDP, we use the RDP accountant in TensorFlow Privacy. 

\subsection{Cluster-based Dependencies}\label{sec:cluster}

\begin{figure*}
	\centering
		\includegraphics[width=\linewidth]{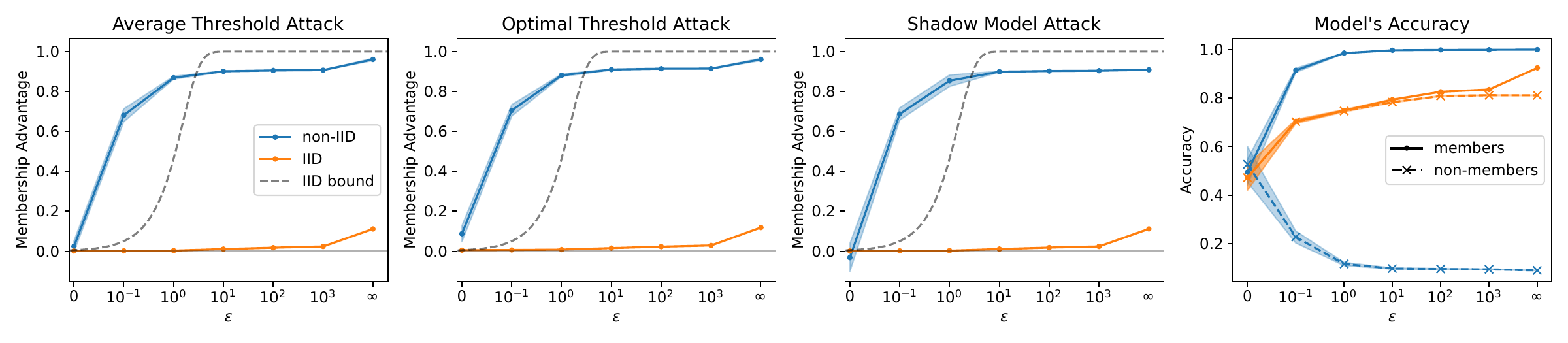}
	\caption{Results in \adult~dataset with cluster-based dependencies.}\label{fig:adult-cluster}
\end{figure*}

\begin{figure}
	\centering
		\includegraphics[width=\linewidth]{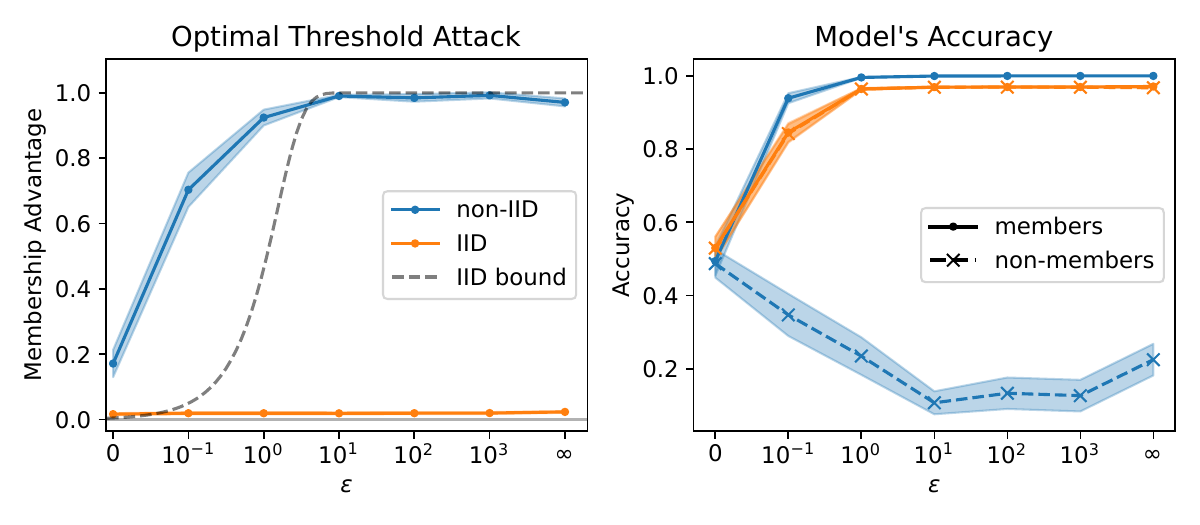}
	\caption{Results in \compas~dataset with cluster-based dependencies.}\label{fig:compas-cluster}
\end{figure}

We begin by studying to what extent data dependencies can affect MIAs.
To do this, we artificially create two sets $D_1$ and $D_2$ that have samples whose features lie in different regions of the feature space.
A similar idea was explored in previous work~\cite{nagarajan2021understanding} by adding spurious perturbations to the training data that caused a large discrepancy between training and testing accuracy.
Instead of modifying existing attributes, we take a real-world dataset and run a $k$-means clustering algorithm with $k=2$ to find clusters of samples that lie in different regions of the feature space.
We run the clustering algorithm independently for all the samples of each class, and add all samples in cluster $i$ to the set $D_i$ ($i\in\{1,2\})$.
Intuitively, $D_1$ and $D_2$ represent mixture components where the characteristics used to distinguish the classes are significantly different.

We use \adult~with $n=10\,000$ and \compas~with $n=2\,000$ in this experiment.
We evaluate what happens when we use $D_1$ to sample the members and $D_2$ for the non-members, and vice-versa.
We use leftover samples in \adult~(up to $n$ per cluster) to train the shadow models for the \shokri~(see Table~\ref{tab:setup}).
We do not evaluate the \shokri~for the \compas~dataset, since we do not have enough samples to provide the adversary.

We run the clustering algorithm once to define $D_1$ and $D_2$, and then repeat the creation of member and non-member sets, the training, and the attack evaluation 50 times.
For each of these experiments, we also run the IID counterpart as explained above, i.e., by randomly mixing the member and non-member sets and repeating the training and attack evaluation.

Figure~\ref{fig:adult-cluster} shows the average advantage of the \average, \optimal, and \shokri~(shaded areas are the $95\%$ confidence intervals for the mean) versus the privacy level $\epsilon$.
We only show the results for the case when members are sampled from $D_1$ and non-members from $D_2$ (the results for the opposite scenario were very similar).
For the value $\epsilon=\infty$, we simply use non-private SGD.
For reference, we also plot our membership advantage bound \eqref{eq:newbound},
which is only guaranteed to hold in the IID scenario.
The rightmost plot includes the average classification accuracy of the target model over the member and non-member sets.

We see that these cluster-based dependencies significantly increase the adversary's performance compared to the IID case.
In general, we observe that the attack performance increases as the privacy level decreases (higher $\epsilon$).
For $\epsilon\geq 1$ we see that the attack performs alarmingly well, close to perfect advantage.
Our results also confirm that the advantage bound does not hold in non-IID cases.

The classification accuracy gives insight into the performance of the threshold attacks.
Recall that these attacks classify all points with logloss \emph{below} a threshold as members. 
The high discrepancy between the loss of members and non-member samples makes it easier for these attacks to identify their membership correctly.
Another interesting observation is that, in Figure~\ref{fig:adult-cluster}, the classification accuracy remains constant after $\epsilon> 1$ but the non-IID advantage of the \yeom~further increases when $\epsilon=\infty$.
This corroborates the observation of Yeom et al.~\cite{yeom2017privacy} that overfitting is not the only factor at play in the success of the attack.

In Figure~\ref{fig:compas-cluster}, we show the results of the \optimal~for the \compas~dataset. 
These results confirm our findings in the \adult~dataset, and even reach an almost-perfect advantage ($\approx 1$) for $\epsilon=10$.
To summarize, we have shown that, using off-the-shelf attacks and datasets, certain dependencies among members/non-members can significantly increase the performance of MIAs.
The advantage we observe not only crosses the DP bound for the IID case, but can reach the maximum value of $1$ in the worst case.

\subsection{Dependencies Caused by an Attribute Bias}\label{sec:att}
We have seen that the cluster-based dependencies can yield extremely high levels of membership advantage.
However, a dependence this strong might not occur in practice.
In this section, we consider a more realistic and interpretable dependency, based on an attribute in the dataset, that we call an \emph{attribute} bias.
We build two experiments using the \adult~dataset: one using the gender attribute (there are only two genders in this dataset, \texttt{`Male'} and \texttt{`Female'}), and another one using the education level (we consider that samples with at least \texttt{`some-college'} education are ``high education'', and all the others are ``low education''). 
For each experiment, we use member and non-member sets with $n=10\,000$ samples, and a bias value $p\in[0, 1]$ representing the proportion of the training set that has the particular attribute value.

We generate the sets $D_1$ and $D_2$ as follows.
We use the gender attribute as an example, but the set generation works similarly for the education attribute.
First, we generate $D_1$ by taking \mbox{$\lceil p\cdot n\rceil$} samples
at random from the dataset with the gender attribute \texttt{`Male'}, and the remaining samples have the \texttt{`Female'} attribute.
We build $D_2$ by taking an even number of samples with each gender attribute value (i.e., $5\,000$ \texttt{`Male'} and $5\,000$ \texttt{`Female'}).
We delete the gender attribute from all samples after generating the sets, so that the adversary cannot use the gender directly in identifying members and non-members.
As before, we use remaining data samples to train the shadow models for the \shokri, ensuring we do not give the adversary more than $n$ samples per attribute value.
Since our goal here is to investigate training set bias, we train with the biased set ($D_1$) and vary the bias level $p$.
The non-members ($D_2$) have an equal number of samples from each attribute value.
We repeat the set generation, training, and evaluation process 50 times for each of the values of the bias $p$ that we test.

\begin{figure*}
	\centering
		\includegraphics[width=\linewidth]{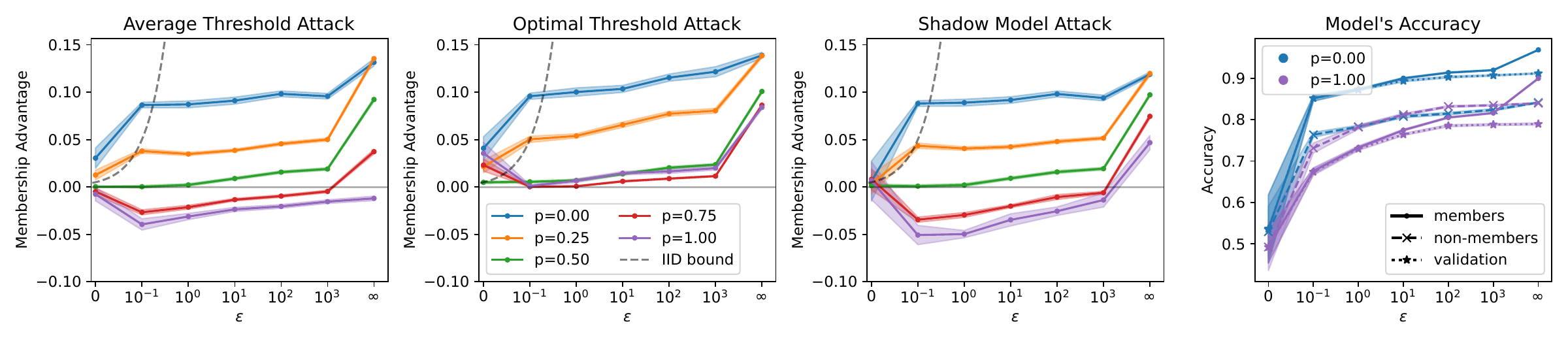}
	\caption{Results in \adult~dataset with gender bias ($p$ determines the proportion of males in the member set).}\label{fig:adult-gender}
\end{figure*}

Figure~\ref{fig:adult-gender} shows the average advantage of the \average, \optimal, and \shokri~for the gender attribute bias, versus the privacy level $\epsilon$.
Each colour represents the performance for a different training set bias, as shown in the legend.
We show only some of the values of the bias $p$ for the model's classification accuracy, for readability.

\begin{figure}
	\centering
		\includegraphics[width=\linewidth]{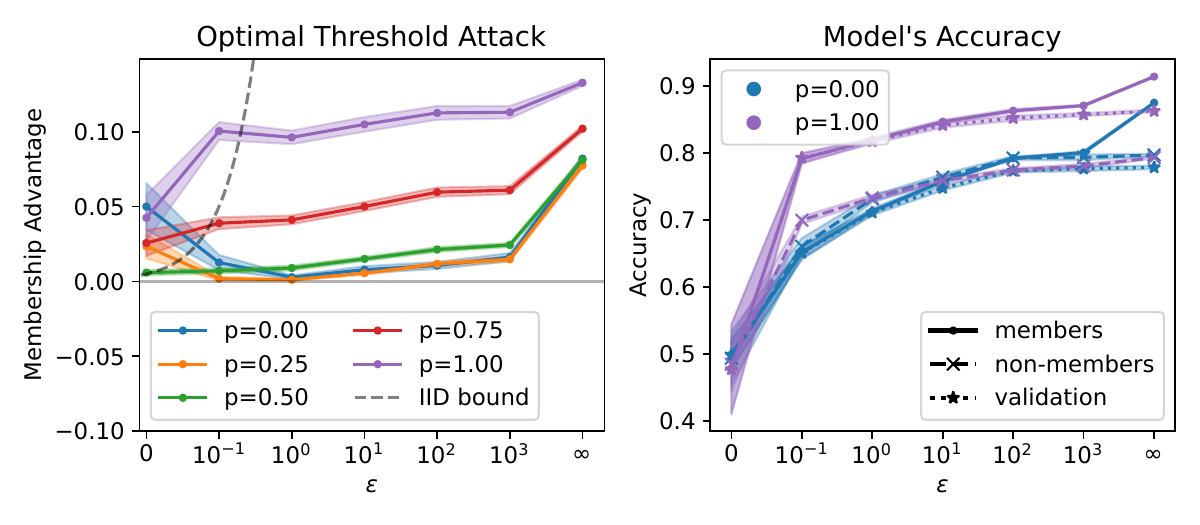}
	\caption{Results in \adult~dataset with education bias ($p$ determines the proportion of individuals with low education in the member set).}\label{fig:adult-numedu}
\end{figure}

We note that, when $p=0.5$, both the member and non-member sets have the same number of samples from each attribute value, so this can be considered an instantiation of the IID experiment.
Training bias ($p\neq 0.5$) can significantly increase the membership advantage, and can even have more impact than the privacy level $\epsilon$.
In fact, we see that the adversary achieves similar advantage in a model trained on unbiased data ($50\%$ males) without privacy ($\epsilon=\infty$), than in a model trained on highly biased data ($0\%$ males) with a high privacy level ($\epsilon=0.1$).
For bias levels $p\geq 0.75$, the membership advantage can be negative.
The classification accuracy gives insight into this effect for the \average. 
When the training set consists of all male samples ($p=1.00$), we see that the classification accuracy for non-members is higher than for members.
This means that the classification task (salary prediction) is likely easier for predominately female samples than is on male samples (even when trained on males only).
This misleads the \yeom~(which classifies all samples with low loss as members) causing it to perform poorly (negative $\Adv$) when the training set is male-dominant.

In this experiment, we also consider the classification accuracy of a validation set that has the same bias as the training set.
That is, we sample $2\,500$ data points that are disjoint from the training set with the same attribute bias as the training set.
We see that the target models show very similar member and validation set accuracy, especially for low values of $\epsilon$.
This suggests that, to a data owner who only has access to biased data, the model appears to generalize well.
Even if the data owner runs a tool such as ML Privacy Meter~\cite{murakonda2020ml} to check their model, they might underestimate the vulnerability of their model.

We plot the results using the education split for the \optimal~only in Figure~\ref{fig:adult-numedu}.
The results are qualitatively similar to the ones with the gender split.
In this case, predicting the income of individuals with low education levels is generally easier than on individuals with high education.
Even when training with high-education samples only ($p=0$), the model's accuracy over non-members (samples with low education) is higher than in members, which causes the MIAs to perform poorly for low values of $p$.
However, when $p$ is large (majority of low-education samples in the training set), we see that the MIAs perform better, breaking the DP bound for the IID case.

\subsection{Inherent Dataset Dependencies}\label{sec:dataset}
We have seen that dependencies play a significant role in the success of membership inference attacks.
A natural question to ask is whether such dependencies occur in practice.
While there are no public examples of MIAs in the wild, we study naturally occurring dependencies found in off-the-shelf datasets.
We do this by either finding different datasets for a similar classification problem or splitting existing datasets based on an attribute that would define a training set in practice, such as the location/institution where the samples were collected.

\paragraph{Hospital data: heart condition detection}
For our first dataset-based dependency experiment, we use the \heart~datasets, which consist of four different datasets (Cleveland, Hungary, Switzerland, and the VA Long Beach), each containing the samples from one hospital/institution.
The classification problem is detecting whether or not a patient has a heart condition (the classification problem is binary).
We use all the samples from each database to define the $\{D_i\}_{i\in[4]}$ sets, and evaluate the performance of MIAs when we train with one database and all the samples from the other three databases are non-members.

Figure~\ref{fig:heart} shows the results when training with each of the hospitals.
We see that there is certain property present within each hospital, since the membership advantage when we train with one hospital dataset (non-IID) is significantly higher than when we shuffle all datasets (IID).
We note that the advantage in the IID case slightly breaks the bound.
We posit that this is because the RDP accountant assumes batches are sampled without replacement and each batch is sampled independently.
In practice, and in our implementation, batches are non-overlapping.
This difference in sampling would be particularly noticeable when the training set is small (as overlaps are more probable), which explains why the advantage in the IID case is slightly above the bound.

\begin{figure*}
	\centering
		\includegraphics[width=\linewidth]{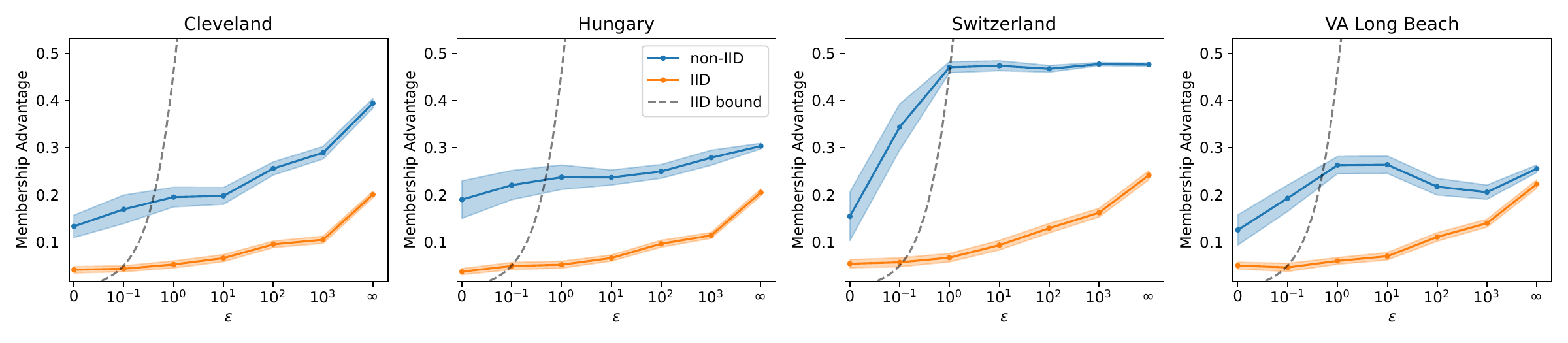}
	\caption{Performance of the optimal threshold attack in the \heart~dataset when members belong to a particular database (hospital/institution), and non-members are taken from all other databases.}\label{fig:heart}
\end{figure*}

\paragraph{School data: students proficiency detection}
Next, we consider the \students~dataset, which contains data about students from two Portuguese schools (Gabriel Pereira and Mousinho da Silveira), and the classification problem is to predict whether the student passes the course.
We use all $423$ samples from Gabriel Pereira school ($D_1$) as the training set, and all $226$ samples from Moushinho da Silveira school ($D_2$) as non-members.
Figure~\ref{fig:students} shows the results.
As before, we see that there is a common property in the students of each school that helps MIAs succeed with higher probability than in the case where we shuffle the data from both schools (IID).

\begin{figure}
	\centering
		\includegraphics[width=\linewidth]{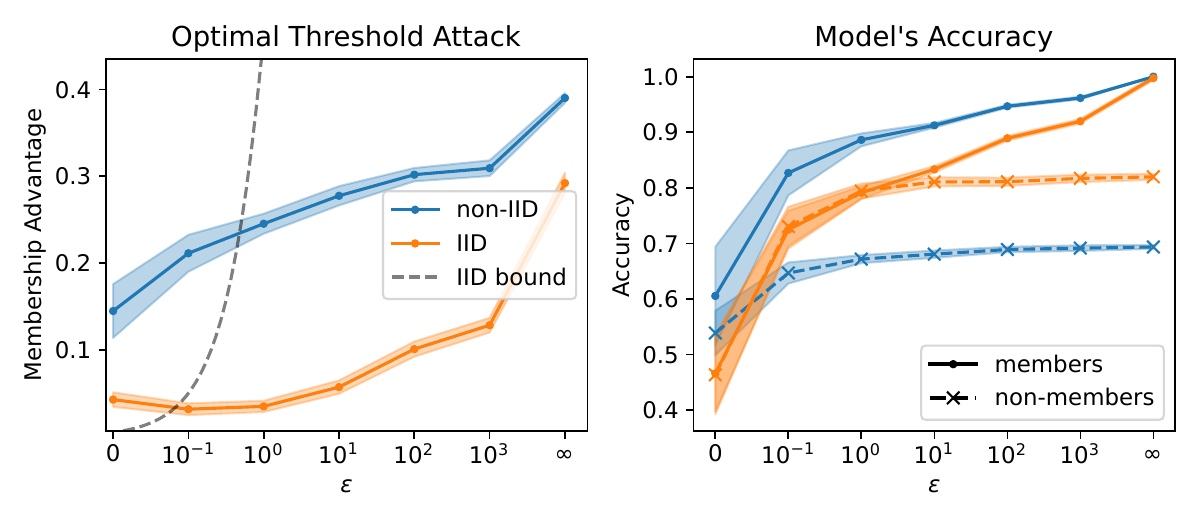}
	\caption{Results in \students~dataset when members are from the Gabriel Pereira school, and non-members from Mousinho da Silveira school.}\label{fig:students}
\end{figure}

\paragraph{Hospital data: large dataset with a multiclass problem}
Next, we use the \texas~dataset, which is a larger dataset ($350\,280$ samples) with the 100-class problem proposed by Shokri et al.~\cite{shokri2016membership}.
Each sample contains an attribute specifying the health region where it was collected.
We assign samples from health region number 3 to the set $D_1$, and assign all other samples to $D_2$.
For each run of this experiment, we randomly sample $n=10\,000$ samples from $D_1$ to use as members/training, and use the same number of samples from $D_2$ as non-members (we restrict the sizes for computational reasons).
We give $n$ leftover samples from each set to the adversary to train the shadow model attack.
We note that such split could naturally occur in practice (e.g., the local authorities of health region 3 decide to train a machine learning model using data for all the hospitals in that region).
Figure~\ref{fig:texas} shows the results.
In this case, the difference between the non-IID case (training with health region 3) and IID case (training with samples from all health regions) is smaller than in previous experiments.
We can see that the classification accuracy is only slightly higher for members in the non-IID than in the IID case (the opposite is true for non-members), which explains the more modest advantage increase due to the training set dependencies.
However, the relative increase in advantage between the IID and non-IID experiments is significant, especially in the threshold attack when $\epsilon\in[1,100]$.

\begin{figure*}
	\centering
		\includegraphics[width=\linewidth]{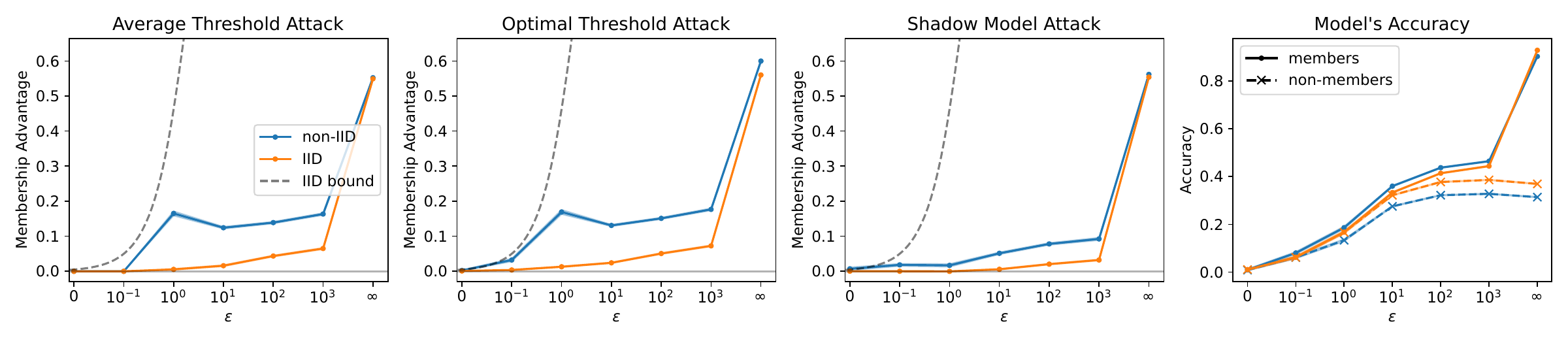}
	\caption{Results in \texas~dataset when members are from hospitals in region 3, and non-members are from any other region.}
	\label{fig:texas}
\end{figure*}

\paragraph{Census data: dataset curation dependencies}
Finally, we consider the dependencies that stem from techniques a researcher or institution follows when curating a dataset.
We use both \adult~and \census~datasets, which have been extracted from \emph{the 1994 US Census data}, but were curated into ML training data by different researchers.
Specifically, each research group extracted a different number of records: the \adult~dataset consists of $48\,842$ records extracted using certain binary rules to filter out records, and the \census~ dataset contains $299\,285$ records, but does not specify any filters.
Both datasets consider the classification task of predicting if an individual's income is above $50\,000$ USD (binary classification).
We take 10 attributes that both datasets have in common. 
The set $D_1$ is the data from \census, and $D_2$ contains the samples from \adult~(we remove duplicates so that $D_1\cap D_2 = \emptyset$).
As in the previous experiment, we set $n=10\,000$ and sample $n$ data points from $D_1$ as members and $n$ from $D_2$ as non-members.\footnote{We remark that, since both of these datasets are extracted from the same source (the US Census), there a possibility that multiple unique records can describe the same individual. However, we consider membership inference at the data point level for this experiment.}
We further draw $n$ samples from each set to train the shadow model attack.
We show the results in Figure~\ref{fig:census}.

We see that, even though the census data comes from the same source, the curation of these datasets introduces dependencies that can certainly help MIAs.
The performance across all attacks is considerably higher for the non-IID setting, and DP-SGD does not seem to have a big effect at preventing this leakage. We note that the noise that DP-SGD adds during training decreases with $n$; this explains why, in this dataset, even small values of $\epsilon$ are not effective at stopping the MIAs in presence of data dependencies.

\begin{figure*}
	\centering
		\includegraphics[width=\linewidth]{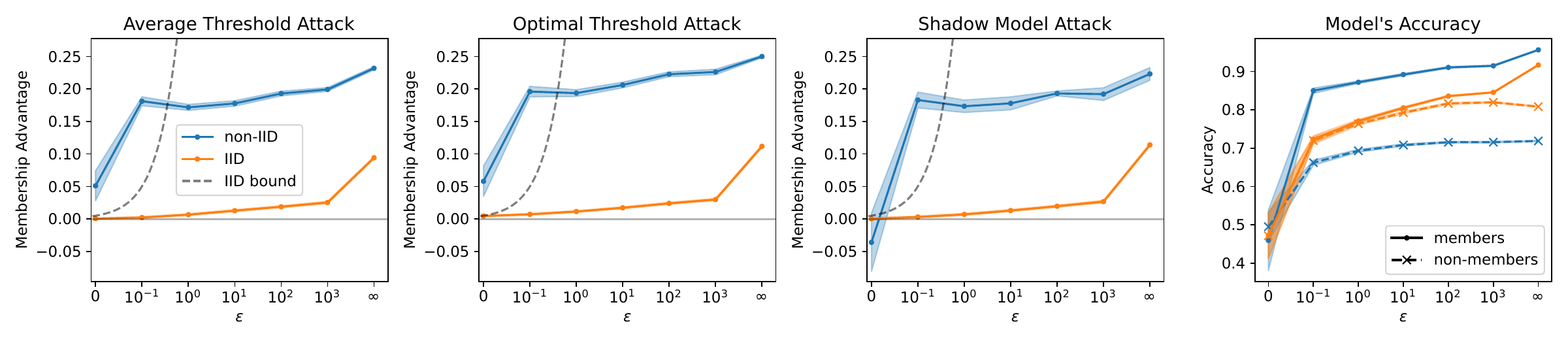}
	\caption{Results when training a model with \census, where non-members are from \adult~dataset. }\label{fig:census}
\end{figure*}

%% file: discussion.tex
\section{Discussion: Connection to Property Inference}
\label{sec:disc}
In this paper, we investigate the effects of data dependencies on membership inference attacks.
We show that \emph{statistical differences} between members and non-members can improve the performance of existing off-the-shelf MIAs.
One can interpret these data dependencies as a \emph{property} that allows distinguishing members from non-members.
This property can be tangible, such an existing attribute that already distinguishes members from non-members, or a non-interpretable and arbitrarily complex combination of attributes.

Property Inference Attacks (PIAs)~\cite{ateniese2013hacking,Ganju_2018,mahloujifar2022property} aim at revealing a property of the training set that the data owner might have wanted to keep secret.
In contrast, recall that MIAs aim to infer whether an individual sample is a member of the training set.
Given their definitions, MIAs and PIAs are very different attacks, with different goals altogether.
However, by defining the data dependencies among members as a property, one could devise a \emph{new MIA} that first runs a PIA to reveal a property of the training set and then checks the presence/absence of this property in the target sample to decide on its membership.
This \emph{MIA-via-PIA} draws connections between membership and property inference in the non-IID case: a successful PIA could lead to a successful MIA.
In this section, we first discuss the feasibility of an MIA-via-PIA, then we discuss the privacy implications of this attack and, finally, explain the differences between an MIA-via-PIA and the MIAs we evaluated in Section~\ref{sec:eval}.

\paragraph{Feasibility of MIA-via-PIAs}
The first step in an MIA-via-PIA is finding a property that separates members from non-members.
PIAs receive a description of the property to be inferred as an input, e.g., as two complementary hypotheses~\cite{ateniese2013hacking,Ganju_2018} or as a binary function of the samples' attributes~\cite{mahloujifar2022property}.
Generally speaking, finding the right property description that separates members from non-members is challenging.
In the theoretical MM membership experiment (Algorithm~\ref{alg:ExpJoint}), the adversary has access to $\distn$, a joint distribution made up of mixture components.
Thus, in this setting, the adversary has a description of a set of properties $[K]$ representing the various mixture components.
The adversary could train a meta-classifier~\cite{Ganju_2018} that receives a trained model as input, and infers which mixture component $\dist_{\hat{k}}$ ($\hat{k}\in[K]$) the training samples followed.
However, in practice, this property inference is likely to be unsuccessful for a number of reasons.
First, the assumption that the adversary receives $\distn$ is very strong and unlikely to occur in practice.
Second, the number of mixture components $K$ is likely very large or might not even be a discrete set (e.g., in the case of hospital data, $K$ could be the total number of ``possible hospital distributions'' that exist, which might not even be finite).
This further increases the difficulty of the property inference (estimation of $\hat{k}$).

The second step in an MIA-via-PIA is to use the property $\hat{k}$ to attempt to infer the membership of a sample.
For example, given the target sample $z$, the attacker could check its likelihood of being a member by checking if $\Pr(z|z\sim \dist_{\hat{k}}) > \Pr(z|z\sim \dist_{k'})$ for all $k'\neq \hat{k}$.
Performing this step can be highly difficult for different reasons.
First, we note that MIAs are a \emph{statistical game}: even when the adversary correctly estimates the property, it is not always possible to correctly guess the membership of the target sample.\footnote{Note that 100\% accuracy in membership inference is not possible, since the number of possible datasets typically exceeds the number of possible model parameters, and thus at least two different datasets may result in the same model, by the pigeonhole principle.}
Second, the mixture components' support is very likely overlapping such that, given a particular data point $z$, there might be several mixture components from which the sample might have come from (i.e., $\Pr(z|z\sim \dist_{\hat{k}})\approx\Pr(z|z\sim \dist_{k'})$ for at least one $k'\neq k$).
In these cases, the MIA-via-PIA will likely fail.
Note that this will often be the case when $K$ is large, which is likely to happen in practice, as we argue above.

\paragraph{Privacy Violations}
Next, we discuss the meaning of privacy violations under an MIA-via-PIA. 
While PIAs violate the privacy of the training set as a whole, MIAs violate the privacy of individual members.
This questions whether the privacy violation of a successful MIA-via-PIA is an individual or a group privacy violation.
The MIA-via-PIA approach decides that any sample $z$ that has the property is a member, regardless of whether or not that sample was truly in the training set.
In that case, it would seem that correctly guessing membership is independent of the actual membership of the sample, and thus is not an individual privacy violation but a group privacy violation.
However, we want to reiterate that MIA is a \emph{statistical game}.
In situations where there is a property that accurately distinguishes members from non-members, an MIA-via-PIA is a smart strategy from a Bayesian inference point of view, that could increase the adversary's odds of winning the statistical membership game.
Leveraging the property \emph{does not change the meaning} of the privacy violation from an individual to a group violation.
We note that leveraging dependencies to improve the accuracy of statistical attacks is common in many privacy areas, such as location privacy (spatio-temporal correlations between a user's locations can help the adversary's estimation of these user's individual locations~\cite{shokri2016privacy}), database privacy (correlations between client queries aid finding out the underlying secret keyword of each individual query~\cite{oya2022ihop}), etc.
These privacy violations are generally considered as individual, and not group, violations.

\paragraph{MIAs under data dependencies}
While we have argued that an MIA-via-PIA is a valid MIA that can be successful, a PIA is not necessary to build a successful MIA under statistical dependencies.
Our theoretical experiments follow prior work~\cite{yeom2017privacy} and consider an informed adversary with knowledge of $\distn$.
We have shown, theoretically, that DP bounds do not hold against such adversary (Section~\ref{sec:boundsdonothold}).
However, in practice, assuming the adversary knows $\distn$ might not be realistic.
Our empirical evaluation considers a particular case of the theoretical membership experiment where the adversary simply does not use (and does not know) $\distn$.
We evaluate popular off-the-shelf attacks and see that these attacks achieve high accuracy under data dependencies, exceeding the DP bounds that hold for the IID scenario.
This shows that knowledge of $\distn$ (which could enable a MIA-via-PIA) is not necessary to exploit the data dependencies and break the IID DP bounds.

Regardless of whether one uses an MIA-via-PIA in non-IID settings, it is a fact that data dependencies affect the performance of off-the-shelf MIAs (our empirical evaluation demonstrates this) and thus MIAs in the wild will inevitably (and perhaps unintentionally) also be affected by this.
We believe that MIA advantage in the wild will be closer to the experiments that we have seen in this paper, and we maintain that future evaluations of MIAs should consider realistic member/non-member splits, since we have seen that random shuffling underestimates the privacy risks of membership inference.
Finding techniques that help mitigate the leakage that stems from data dependencies, which has been ignored by past work, is an interesting topic for future work.

%% file: related.tex
\section{Related Work}
Chatzikokolakis et al.~\cite{Chatzikokolakis_2020} provide a similar theoretical bound for the case of pure differential privacy. 
However, since DP-SGD uses the moments' accountant, which gives approximate DP, our bound is more useful in practice.
We remark that the goal of this paper is not to be competitive with state-of-the-art attacks in MIA~\cite{jayaraman2020revisiting, song2021systematic,sablayrolles2019white, long2020pragmatic, watson_21, Ye_21, carlini2022membership}. 
Instead, we show the amplification effect of data dependencies on existing, well-known attacks by Shokri et al.~\cite{shokri2016membership} and Yeom et al.~\cite{yeom2017privacy}.
Showing that a single MIA breaks the upper bound for a single attack is enough to prove our point.

There have been several works that investigate the performance of MIAs under different assumptions.
Nasr et al.~\cite{nasr_comprehensive_2018} provide a comprehensive privacy analysis of MIAs under several white-box conditions and design an inference attack that targets vulnerabilities in the stochastic gradient descent algorithm.
Recent work by Jagielski et al.~\cite{jagielski2020auditing} improves MIA success by developing a novel data poisoning attack.
Nasr et al.~\cite{nasr2021adversary} show experimentally that the strong adversary (see Alg.~\ref{alg:nasr}) can reach the theoretical DP upper bounds. 
These works all show stronger attack success than previously reported. However, they achieve this performance by giving the adversary increased (often unrealistic) capabilities while only considering the case of IID data. We instead focus on the weaker black box model of attacker capabilities and vary the data distribution (to be more realistic) to strengthen the attacks.

Kulynych et al.~\cite{kulynych2022disparate} demonstrate that data samples from underrepresented communities are more vulnerable to MIAs.
Bagdasaryan et al.~\cite{bagdasaryan2019differential} find that DP-SGD results in a disproportionate amount of accuracy loss for vulnerable groups.
These works complement our findings that MIAs are stronger when the training data has dependencies.
Xiong et al.~\cite{Xiong2022} study the risks of non-IID data in the federated learning setting. 
Hu et al.~\cite{hu2021source} design a new attack to uncover the source of data points (such as which hospital contributed the data point) in the federated setting. 
This work complements ours by showing an increase in attack performance when considering clients with different data distributions. 
Similarly, Gomrokchi et al.~\cite{gomrokchi2021did} study the problem of MIAs against reinforcement learning and find that temporal correlation between training trajectories plays a major role in attack success.

We recall that property inference attacks aim to discover hidden properties in a target classifier's training dataset, such as the proportion of members in the dataset that are students~\cite{ateniese2013hacking}.
It is well known that DP does not protect against property inference attacks~\cite{ateniese2013hacking,mahloujifar2022property}.
As discussed in Section~\ref{sec:disc}, our goal is not to conduct a property inference attack or be competitive with state-of-the-art attacks in this field~\cite{ateniese2013hacking, Ganju_2018, mahloujifar2022property, zhang2021leakage,parisot2021property}.
Instead, we show the effects of datasets that are non-IID on MIA attacks using properties as one way to introduce a dependency in the datasets.

While our work is the first to investigate the effects of dependent data on MIAs against Neural Networks trained with DP-SGD, other works have studied the limitations of DP under data dependencies, as summarized by Tschantz et al.~\cite{tschantz_sok_2020}. 
Kifer and Machanavajjhala~\cite{kifer_no_2011} claim that DP mechanisms do not always limit participation inference, and later propose the Pufferfish framework~\cite{pufferfish}, which they use to show that data correlations can cause additional information leakage compared to the IID case.
Li et al.~\cite{li2013membership} explain that DP implies membership privacy only when data is mutually independent.
Liu et al.~\cite{liu2016dependence} empirically study the effect of data dependencies on inference attacks against the Laplace mechanism. 
They show that an adversary with knowledge of these dependencies can violate DP bounds that hold in the IID scenario.
Although our conclusions are similar (besides the fact that we work in the ML domain), an important difference between this work and ours is that the attacks we evaluate do not use knowledge about these data dependencies.
Instead, we show that existing attacks~\cite{shokri2016membership, yeom2017privacy}, which were not designed to exploit data dependencies, also benefit from them.

%% file: conclusion.tex
\section{Conclusion}
We show that current MIA experiments and evaluations hinge on the restrictive assumption that members and non-members are IID data samples.
We studied the performance of MIAs when there are statistical dependencies among the training set samples, and found that they pose a far greater threat than previously reported.
We show that current state-of-the-art MIAs can achieve near-optimal performance when we introduce artificial dataset dependencies, but even naturally occurring dependencies can increase MIA performance compared to the IID scenario.
We show theoretically and empirically that, when members and non-members are not IID, the previous theoretical bounds of DP do not apply. 
Our work demonstrates that data dependencies should be taken into account when studying MIA performance, as they are a realistic assumption that, if ignored, can lead to a significant underestimation of the privacy risk that MIAs pose.

%% file: appendix.tex
\subsection{Proof of Theorem~\ref{theo:newIID}}
\label{sec:app}

We prove that the bound in Theorem~\ref{theo:new} holds when members and non-members are \emph{statistically exchangeable}.
Then, we use this to show that this bound holds in  Yeom et al.'s~\cite{yeom2017privacy} membership experiment ($\ExpIID$), even when members and non-members are sampled without replacement. 

\begin{defn}[Statistical Exchangeability in MIAs]
	\label{def:ex}
	A joint distribution $\distn$ for $n$ member samples $\{z_1,\dots,z_n\}$ and a non-member sample $z_{n+1}$ meets statistical exchangeability if the joint distribution of the samples $\{z_1,\dots,z_{n+1}\}\sim\distn$ does not depend on the sample order in the sequence.
	Formally, for any permutation $\sigma: [n+1]\to[n+1]$,
	\begin{equation}
		\Pr(z_1,\dots,z_{n+1})=\Pr(z_{\sigma(1)},\dots,z_{\sigma(n+1)}).
	\end{equation}
\end{defn}

We define a new membership experiment that generalizes $\ExpIID$ to the case where members and non-members are statistically exchangeable:

\begin{algorithm}
	\begin{algorithmic}[1]
	\Procedure{$\ExpEX$}{$\Att, A, n, \distn$}  \Comment{$\distn$ exchangeable.}
	\State Sample $\{z_1,z_2,\dots,z_{n},z'\}\sim\distn$;
	\State Choose $i\sim[n]$ uniformly at random;
	\State Set $z=z_i$ and $\tilde{S}=\{z_k\}_{k\neq i}$;
	\State Choose $b\sim\{0,1\}$ uniformly at random; \label{line:att0}
	\If{$b=0$} \label{line:if0}
		\State Train $a=A(\tilde{S}\cup\{z\})$;
	\Else
		\State Train $a=A(\tilde{S}\cup\{z'\})$;
	\EndIf \label{line:if1}
	\State Return 1 if $\Att(z, a, n, A, \distn)=b$; else 0. \label{line:att1}
	\EndProcedure
	\end{algorithmic}
	\caption{Exchangeable Membership Experiment}
	\label{alg:ExpEX}
\end{algorithm}

We first prove the following result:
\begin{theorem}[New Membership Advantage Bound]
	\label{theo:newEX}
	Let $A$ be an ($\epsilon$,$\delta$)-DP learning algorithm. 
	Then, for all attacks $\Att$, training set sizes $n$, and statistically exchangeable joint distributions $\distn$, the membership advantage in $\ExpEX$ satisfies
	\begin{equation}
	 \Adv(\Att,A,n,\dist)\leq (e^\epsilon-1+2\delta)/(e^\epsilon+1)\,.
	\end{equation}
\end{theorem}

\begin{proof}
First, consider a variation of $\ExpEX$ that replaces the call to the attack $\Att(z,a,n,A,\dist)$ in $\ExpEX$ (Alg.~\ref{alg:ExpEX}, line~\ref{line:att1}) with a more informed attack that receives both $z$ and $z'$, as well as $\tilde{S}\doteq S\setminus \{z\}$.
We call this new experiment $\ExpAltt$.
We write the call to the attack as $\Att(a, z, z', \tilde{S}, A)$: we disregard $n$ as an input, since it can be inferred from $\tilde{S}$, and $\distn$, since the adversary receives all variables sampled from it and, due to the statistical exchangeability property, $\Pr(\tilde{S},z,z')=\Pr(\tilde{S},z',z)$, and thus $\distn$ does not provide any information that helps distinguish a member from a non-member.
Let $\AdvAlt$ (resp.~$\AdvAltt$) be the maximum advantage of any attack in $\ExpEX$ (resp.~$\ExpAltt$).
Notice that, since the attack in $\ExpAltt$ receives strictly more information than the attack in $\ExpEX$, then $\AdvAlt\leq\AdvAltt$.

Finally, notice that the process between lines~\ref{line:att0} and \ref{line:att1} in Algorithm~\ref{alg:ExpEX} (resp.~$\ExpAltt$) is exactly the same as $\ExpSA$ in Algorithm~\ref{alg:nasr}.
In other words, we can rewrite $\ExpAltt$ as a new experiment, $\ExpAlttt$, shown in Algorithm~\ref{alg:ExpAlttt}.
We have proven that the membership advantage of $\ExpSA$ is upper-bounded by \eqref{eq:newbound}, regardless of how $z$, $z'$, and $\tilde{S}$ were generated.
Therefore, the same bound holds for $\ExpAlttt$.
Since $\ExpAlttt$ is equivalent to $\ExpAltt$ and $\AdvAlt\leq\AdvAltt$, it holds that our bound in \eqref{eq:newbound} also holds for $\ExpEX$.
\end{proof}

We note that $\ExpIID$ can be written as $\ExpEX$ for the particular case where $\distn=\dist^n$.
This proves Theorem~\ref{theo:newIID}.
The variation of $\ExpIID$ that samples without replacement also results in a joint distribution $\distn$ that satisfies statistical exchangeability, and thus the bound also holds in this case.
Finally, we note that the joint distribution $\distn$ in the mixture model experiment $\ExpNoIID$ does \emph{not} satisfy the statistical exchangeability property, and thus we cannot prove the bound holds following this approach (in Section~\ref{sec:gen} we show why such a bound does not hold generally in the non-IID case).


\begin{algorithm}
	\begin{algorithmic}[1]
		\Procedure{$\ExpAlttt$}{$\Att, A, n, \distn$} \Comment{$\distn$ exchangeable.}
		\State Sample $\{z_1,z_2,\dots,z_{n},z'\}\sim\distn$; 
		\State Choose $i\sim[n]$ uniformly at random;
		\State Set $z=z_i$ and $\tilde{S}=\{z_k\}_{k\neq i}$;
		\State Return $\ExpSA(\Att^*,A,\tilde{S},z,z')$.
		\EndProcedure
	\end{algorithmic}
	\caption{Exchangeable Membership Experiment (v2)}
	\label{alg:ExpAlttt}
\end{algorithm}

\subsection{Tightest Bound}
\label{app:tightest}
We prove that our new bound ($(e^\epsilon-1+2\delta)/(e^\epsilon+1)$) is tighter than the bound by Yeom et al.~\cite{yeom2017privacy} ($e^\epsilon-1$) and Erlingsson et al.~\cite{erlingsson_that_2020} ($1 - e^{-\epsilon} ( 1 - \delta)$).
First, we note that Yeom et al's bound assumes $\delta=0$.
In that case, Erlingsson et al.'s bound is clearly tighter:
$1 - e^{-\epsilon}=({e^\epsilon-1})/({e^\epsilon})\leq e^\epsilon - 1$.

Thus, to prove our bound is the tightest of the three, it suffices to prove it is tighter than Erlingsson et al.'s bound, i.e.,
\begin{equation*}
	\frac{e^\epsilon-1+2\delta}{e^\epsilon+1}\leq 1 - e^{-\epsilon} ( 1 - \delta)
\end{equation*}
To do this, we simply compute the difference between the left and right terms of these inequalities, and perform basic arithmetic operations to show the difference cannot be positive:
\begin{equation*}
	\frac{e^\epsilon-1+2\delta}{e^\epsilon+1}-\left(1 - e^{-\epsilon} ( 1 - \delta)\right)=\frac{(\delta-1)(e^{-\epsilon}+1)}{e^\epsilon+1}\leq0
\end{equation*}
The last inequality comes from the fact that the terms $(e^{-\epsilon}+1)$ and $(e^\epsilon+1)$ are strictly positive, while $(\delta-1)$ is non-positive (since $\delta\leq 1$). \qed